\documentclass[superscriptaddress,longbibliography,nofootinbib,notitlepage,rgb,dvipsnames,table,aps,pra,10pt]{revtex4-2}

\usepackage{graphicx}
\usepackage{amsthm}
\usepackage{thmtools}
\usepackage{mathtools}
\usepackage{thm-restate} 
\usepackage{amsmath}
\usepackage{amssymb}
\usepackage{placeins}
\usepackage{multirow}
\usepackage[caption=false]{subfig}
\usepackage{enumitem}
\usepackage[colorlinks]{hyperref}
\hypersetup{
	pdfstartview={FitH},
	pdfnewwindow=true,
	colorlinks=true,
	linkcolor=blue,
	citecolor=blue,
	filecolor=blue,
 	urlcolor=blue}
\usepackage{tikz}
\usetikzlibrary{calc,backgrounds}

\usepackage{floatrow}
\usepackage{pgffor}
\usepackage{url}
\usepackage{hypcap}
\usepackage{dsfont}
\usepackage{comment}
\usepackage{algorithmic}
\usepackage[ruled,lined]{algorithm2e}
\SetKw{KwBy}{by}
\SetKwFor{For}{for}{}{end} 
\SetKwIF{If}{ElseIf}{Else}{if}{}{else if}{else}{end}
\usepackage{soul}
\usepackage{diagbox}
\usepackage{mathrsfs}
\usepackage{xcolor}
\usepackage{enumitem}
\usepackage{booktabs,tabularx,array,xcolor}

\allowdisplaybreaks

\newcommand{\eq}[1]{Eq.~\hyperref[eq:#1]{(\ref*{eq:#1})}}
\renewcommand{\sec}[1]{\hyperref[sec:#1]{Section~\ref*{sec:#1}}}
\newcommand{\app}[1]{\hyperref[app:#1]{Appendix~\ref*{app:#1}}}
\newcommand{\tab}[1]{\hyperref[tab:#1]{Table~\ref*{tab:#1}}}
\newcommand{\fig}[1]{\hyperref[fig:#1]{Figure~\ref*{fig:#1}}}
\newcommand{\figa}[2]{\hyperref[fig:#1]{Figure~\ref*{fig:#1}#2}}
\newcommand{\figx}[2]{\hyperref[fig:#1]{Figure~\ref*{fig:#1}(#2)}}
\newcommand{\thm}[1]{\hyperref[thm:#1]{Theorem~\ref*{thm:#1}}}
\newcommand{\lem}[1]{\hyperref[lem:#1]{Lemma~\ref*{lem:#1}}}
\newcommand{\cor}[1]{\hyperref[cor:#1]{Corollary~\ref*{cor:#1}}}
\newcommand{\defn}[1]{\hyperref[def:#1]{Definition~\ref*{def:#1}}}
\newcommand{\alg}[1]{\hyperref[alg:#1]{Algorithm~\ref*{alg:#1}}}

\newcommand{\boxno}[1]{f_b(#1)}
\newcommand{\Levs}{L}
\newcommand{\Nuc}{\mathcal{\zeta}}
\newcommand{\parent}{P}
\newcommand{\child}{C}
\newcommand{\nearn}{N\!N}
\newcommand{\inter}{I}
\newcommand{\boxwid}{w}
\newcommand{\firstbox}{j_0}
\newcommand{\translate}{\mathcal{T}}
\newcommand{\trunc}{\mathcal{P}}
\newcommand{\near}{\mathrm{near}}
\newcommand{\far}{\mathrm{far}}

\def\bra#1{\mathinner{\langle{#1}|}}
\def\ket#1{\mathinner{|{#1}\rangle}}

\newcommand{\norm}[1]{\left\lVert#1\right\rVert}

\newcommand{\cC}{{\cal C}}
\newcommand\cO{{\cal O}}
\newcommand{\nn}{\nonumber \\}
\newcommand{\vb}[1]{\mathbf{#1}}

\usepackage{color}
\newtheorem{theorem}{Theorem}
\newtheorem{lemma}[theorem]{Lemma}

\usepackage[most]{tcolorbox}

\newtcolorbox[auto counter]{summarybox}[2][]{%
    colframe=blue!75!black,  
    fonttitle=\bfseries,      
    title={Box \thetcbcounter: #2}, 
    label={#1},              
    enhanced,                
    rounded corners,         
    drop shadow              
}

\newcommand{\Google}{\affiliation{Google Quantum AI, Venice, CA 90291, United States}}

\newcommand{\Macquarie}{\affiliation{School of Mathematical and Physical Sciences, Macquarie University, New South Wales 2109, Australia}}

\newcommand{\Sandia}{\affiliation{Quantum Algorithms and Applications Collaboratory, Sandia National Laboratories,
Albuquerque, NM 87185, United States}}

\newcommand{\USyd}{\affiliation{School of Chemistry, University of Sydney, NSW 2006, Australia}}

\newcommand{\Stanford}{\affiliation{Stanford Institute for Theoretical Physics, Stanford University, Stanford, CA 94305,  United States}}

\begin{document}

\title{Quantum simulation of electronic structure via quantum fast multipole method}

\date{\today}

\author{Dominic W.~Berry}
\email{dominic.berry@mq.edu.au}
\Macquarie

\author{Kianna Wan}
\Google
\Stanford

\author{Andrew D. Baczewski}
\Sandia

\author{Elliot C.~Eklund}
\USyd

\author{Arkin Tikku}
\USyd

\author{Ryan Babbush}
\email{ryanbabbush@gmail.com}
\Google

\date{\today}

\begin{abstract}
Here we describe an approach for simulating electronic structure on quantum computers with significantly lower asymptotic complexity than prior work.
The approach uses a real-space first-quantised representation of the molecular Hamiltonian which we propagate using high-order product formulae.
Essential for this low complexity is the use of a technique similar to the fast multipole method for computing the Coulomb operator with $\widetilde{\cal O}(\eta)$ complexity for a simulation with $\eta$ particles.
We show how to modify this algorithm so that it can be implemented on a quantum computer. We ultimately demonstrate an approach with $t(\eta^{4/3}N^{1/3} + \eta^{1/3} N^{2/3} ) (\eta Nt/\epsilon)^{o(1)}$ gate complexity, where $N$ is the number of grid points, $\epsilon$ is target precision, and $t$ is the duration of time evolution.
This is roughly a speedup by ${\cal O}(\eta)$ over most prior algorithms.
We provide lower complexity than all prior work for $N<\eta^7$ (the regime of practical interest), with only first-quantised interaction-picture simulations providing better performance for $N>\eta^7$.
As with the classical fast multipole method, large particle numbers $\eta\gtrsim 10^3$ would be needed to realise this advantage.
\end{abstract}

\maketitle

\section{Introduction}
Solving problems in chemistry is expected to be one of the most promising applications of quantum computing. Much of the early research in this direction focused on developing and improving quantum algorithms to compute energies of molecular systems using the quantum phase estimation algorithm. We often refer to the problem of computing molecular energies as the electronic structure problem.
More recently, an effort has been made to rigorously explore the potential of quantum computing for applications in chemical dynamics where the focus is mainly computing correlation functions, reaction rates, and quantum yields \cite{eklund2026}. Both electronic structure and chemical dynamics rely on efficient quantum algorithms to simulate unitary evolution under the governing Hamiltonian. 
Enabling practical application for a wide range of problems in chemistry depends on developing and optimising algorithms for quantum simulation.
Even more broadly, simulations of materials~\cite{SuPRXQuantum21,berry2024quantum} and degenerate plasmas~\cite{BabbushNC2023,RubinPNAS24} make use of similar techniques, and often require studying systems with many more interacting electrons than is typical in chemistry.
Thus, improvements to the asymptotic scaling of quantum simulation methods with the number of electrons might facilitate their application to not only larger molecular systems but also condensed phases.

The first quantum algorithms designed for quantum simulation were based on product formulae \cite{Lloyd1996}. As the field of quantum algorithms matured, more advanced methods were introduced. These algorithms included techniques based on linear combinations of unitaries \cite{Wiebe2012,BerrySTOC14}, quantum walks \cite{BerryQIC12}, Taylor series \cite{BerryPRL15}, quantum signal processing \cite{QSP}, and block encoding \cite{Low2019hamiltonian}.
A key breakthrough with these more recent methods is that they provide a relatively favourable complexity scaling in the simulation error $\epsilon$, scaling as $\cO(\log(1/\epsilon))$. This is in contrast to product formulae, which scale as $\mathrm{poly}(1/\epsilon)$. However, recent work shows that product formulae can be competitive with logarithmically scaling methods for simulation because they provide low error in practice \cite{TrotComm,Morales2025}.

Simulating the types of Hamiltonians that are relevant in chemistry is made difficult by the fact that all pairwise interactions between charged particles need to be summed over in order to evaluate the Coulomb potential. For a system consisting of $\eta$ electrons, this calculation yields an $\eta^2$ factor in the complexity of the algorithm. However, some approaches to simulation avoid the double sum, reducing this factor to $\cO(\eta)$. For example, when the state of the molecule is encoded as a list of occupied orbitals in a first-quantised representation, it is possible to achieve a linear dependence in $\eta$ by block encoding the Hamiltonian, which does not need the potential to be computed.

For second quantisation, the Coulomb potential may be calculated with a fast Fourier transform, avoiding the double sum \cite{low2019hamiltoniansimulationinteractionpicture}.
However, due to second quantisation the complexity is at least linear in the number of orbitals $N$. In cases where $N\gg\eta$, it is preferable to use first-quantised based approaches to avoid the factor of $N$ in the complexity.
We note that very recently, a first-quantised approach for simulating non-relativistic quantum electrodynamics was proposed that also avoids the $\cO(\eta^2)$ overhead for the Coulomb potential \cite{Stetina2025}.

Simulations based on product formulae in first quantisation \cite{Kassal2008,RubinPNAS24} show promise for better scaling than many of these algorithms, due to the Low \textit{et al.}~\cite{LowPRX2022} bounds on the error in fermionic systems.
A long-standing open question in quantum algorithms has been how to avoid the $\cO(\eta^2)$ overhead for the Coulomb potential in this approach.
In classical algorithms, tree codes (e.g., Barnes-Hut~\cite{Barnes1986}) and the fast multipole method (FMM)~\cite{ROKHLIN1985187,Carrier1988} enable computation of the potential to be performed with complexity $\cO(\eta\,\mathrm{polylog}(\eta)\,\mathrm{polylog}(1/\epsilon))$ or $\cO(\eta\,\mathrm{polylog}(1/\epsilon))$, respectively.
The difficulty with applying these methods in quantum algorithms is that a direct translation of the various classical algorithms would result in data accesses in locations governed by the values in quantum registers.
Each data access has complexity corresponding to the number of data locations, which here is $\eta$, increasing the overall complexity by this factor.
Therefore, this approach results in an overall complexity that is larger than $\cO(\eta^2)$, negating the speedup provided by either tree codes or FMMs.
The FMM approach discussed in Ref.~\cite{Childs2022quantumsimulationof}, which considers the direct translation of the classical algorithm, incurs this extra overhead of $\eta$. In this instance, the overhead results implicitly from the use of QRAM. A detailed discussion of the issues associated with implementing an FMM quantum algorithm is provided by Babbush \emph{et al.}~in Ref.~\cite{BabbushNC2023}.
(Note that the FMM is only relevant to first quantisation, because for second quantisation the approach of Ref.~\cite{low2019hamiltoniansimulationinteractionpicture} may be used.)

Our solution to recover the $\cO(\eta\log\eta)$ scaling of the classical algorithm is related to the implementation of a quantum sort \cite{MergeSort,Beals2013} in which a sorting network is used to make comparisons at fixed locations so that the overhead from accessing quantum data is avoided.
In the FMM algorithm, a tree structure is created, where each successive level of the tree partitions the simulation cell into increasingly smaller boxes, so that at the leaf level each box contains a constant number of particles.
One difficulty involved with our sorting network approach is constructing a subroutine to retrieve information from the boxes so that the appropriate interactions can be evaluated.
We solve this problem using multiple Morton orderings with shifts to ensure that all boxes in the interaction list are within a minimum distance in one order.
The methods developed here to avoid the data-access overhead show promise for application to quantum versions of many other fast summation methods \cite{YESYPENKO2025113707,Angulo2022,Fukuda2022,Wang2021,Gnedin2019,Stenqvist,liang2025}.

The fast multipole method needs large particle numbers in order to show an improvement over direct summation.
Greengard and Rokhlin show a break-even point of about 500 particles for modest accuracy to 5000 particles for high-accuracy calculations \cite{greengard1997new}.
In comparison, Ref.~\cite{RubinPNAS24} involves 1729 electrons for simulation of a deuterium plasma.
Simulation of plasma with larger atomic numbers can be expected to require simulation of thousands of electrons where the quantum fast multipole method becomes relevant.
We caution that the current quantum implementation of the FMM has large overheads, and will require considerable further optimisation to provide comparable performance to the classical FMM.
The focus of this work is on showing that the same scaling as the classical FMM is possible.

In the rest of this paper, we summarise the electronic structure Hamiltonian in Section \ref{sec:struc}, and the FMM algorithm in Section \ref{sec:fmm}. In Section \ref{sec:even}, we explain how to implement a quantum FMM algorithm for the case where the particles are assumed to be evenly distributed. We lift this assumption in Section \ref{sec:adaptive} and explain how to 
implement an adaptive FMM algorithm so that the size of each box is adapted to the local density of particles.
We estimate the total complexity in Section \ref{sec:complex}, and conclude in Section \ref{sec:conc}.

\section{Simulating quantum chemistry in real space first quantisation}
\label{sec:struc}

We begin by considering the simulation of the electronic structure problem defined on a spatial grid in first quantisation. Within the Born-Oppenheimer (BO) approximation, the electronic structure Hamiltonian is defined as
\begin{align}
\label{eq:real_space}
H & = T + U + V + \sum_{l \neq \kappa=1}^\Nuc\frac{q_l q_\kappa}{2\left\|R_l - R_\kappa\right\|}\, , \\
T & = \sum_{i=1}^{\eta} {\rm QFT}_i \left( \sum_{\mathbf{p}\in G} \frac{\left \| \vb{k}_\mathbf{p}\right\|^2}{2} \ket{\mathbf{p}}\!\!\bra{\mathbf{p}}_{i} \right) {\rm QFT}_i^\dagger \, , \\
U & = -\sum_{i=1}^\eta\sum_{l =1}^{\Nuc}  \sum_{\mathbf{p}\in G}\frac{q_l}{\left\|R_l - \vb{r}_\mathbf{p}\right\|} \ket{\mathbf{p}}\!\!\bra{\mathbf{p}}_{i} \, , \\
V & = \sum_{i\neq j=1}^\eta \sum_{\mathbf{p},\mathbf{q}\in G}\frac{1}{2\left\|\vb{r}_\mathbf{p} - \vb{r}_\mathbf{q}\right\|} \ket{\mathbf{p}}\!\!\bra{\mathbf{p}}_{i} \ket{\mathbf{q}}\!\!\bra{\mathbf{q}}_{j} \, ,
\end{align}
where ${\rm QFT}_i$ denotes the standard quantum Fourier transform applied to register $i$ and $\| \cdot \|$ denotes the Euclidean norm. Further, $l$ and $\kappa$ index nuclear degrees of freedom, while $i$ and $j$ index electronic degrees of freedom. The positions of the nuclei and electrons in real space are denoted by $R_l$ and $\vb{r}_\mathbf{p}$, respectively. The atomic numbers of nuclei, which are used to express the nuclear charges, are denoted by $q_l$. Throughout this work, we use $\eta$ and $\Nuc$ to denote the numbers of electrons and nuclei, respectively. We work in atomic units such that $\hbar$, $4\pi\varepsilon_0$, and the mass and charge of the electron are unity.

The simulation cell is specified by its grid points and the associated frequencies in the dual space given by the QFT. In particular, we define
\begin{equation}\label{eq:vardefs}
\vb{r}_\mathbf{p} = \frac{\vb{p} \, \Omega^{1/3}}{N^{1/3}}\, , \qquad \qquad \vb{k}_\mathbf{p} = \frac{2 \pi\, \vb{p}}{\Omega^{1/3}}\, , \qquad \qquad
\mathbf{p} \in G\, , \qquad \qquad G = \left[-\frac{N^{1/3}-1}{2},\frac{N^{1/3}-1}{2}\right]^3\, .
\end{equation}
Here, $\Omega$ is the volume of the simulation cell and $N$ is the number of grid points in the cell. The value of a grid point in real space is $\vb{r}_\mathbf{p}$, and the associated value of the frequency is $\vb{k}_\mathbf{p}$.
For our implementation of the FMM, the position will be encoded by natural numbers in binary for each coordinate.

This is essentially a more precise reformulation of the same representation used by Kassal \emph{et al.}~\cite{Kassal2008} in the first work on quantum simulating chemistry in first quantisation. Our approach is then to perform simulation under the electronic structure Hamiltonian by using high-order product formulae and a split-operator Trotter step, that separately evolves under $T$, and then under $U + V$. The approach we use to evolve under $T$ is essentially the same as that pursued by Kassal \emph{et al.}~\cite{Kassal2008}, resulting in $\widetilde{\cal O}(\eta)$ complexity. In cases where $\Nuc$ is small, it would be appropriate to perform simulation under the operator $U$ using the approach from Kassal \emph{et al.}~\cite{Kassal2008} as well, which results in ${\cal O}(\eta \Nuc)$ complexity. A significant innovation of our work is to compute $V$ using a quantum version of the FMM, then perform time evolution with phase kickback.
This approach can also be used to simulate $U$ in cases with larger $\Nuc$, or for the internuclear potential in non-BO simulations.

Unlike the work of Kassal \emph{et al.}~\cite{Kassal2008}, we propose to perform time evolution using high-order product formulae.
The work of Low \emph{et al.}~\cite{LowPRX2022} bounds the number of Trotter steps required for a split-operator approach in real-space second quantisation.
This result for the number of Trotter steps also holds for real-space first quantisation (see the Supporting Information of Ref.~\cite{RubinPNAS24}, Section V.B.), and is
\begin{align}
\label{eq:trotter_number}
t \left(\frac{\eta^{2/3} N^{1/3}}{\Omega^{1/3}} + \frac{N^{2/3}}{\Omega^{2/3}}\right) \left(\frac{\eta N t}{\Omega \epsilon} \right)^{o\left(1\right)}
\end{align}
for time-evolution duration $t$ and target precision $\epsilon$.
In this expression $o(1)$ indicates powers of the inverse of the order of the product formula, and so may be made arbitrarily small.
Thus, the overall complexity will be given by \eq{trotter_number} times the complexity of the Trotter step.

The complexity of the Trotter step is bottlenecked by the computation of $V$. The na\"ive method for this, described by Kassal \emph{et al.}~\cite{Kassal2008}, has complexity $\widetilde{\cal O}(\eta^2)$. But we will instead show an approach based on the fast multipole method with complexity $\widetilde{\cal O}(\eta)$.
Multiplying the number of Trotter steps by $\widetilde{\cal O}(\eta)$, this will lead to an overall gate complexity
\begin{align}\label{eq:complexity}
t\left(\frac{\eta^{5/3} N^{1/3}}{\Omega^{1/3}} + \frac{\eta N^{2/3}}{\Omega^{2/3}}\right) \left(\frac{\eta N t}{\Omega \epsilon} \right)^{o\left(1\right)}
= t\left(\eta^{4/3}N^{1/3} + \eta^{1/3} N^{2/3}\right) \left(\frac{\eta N t}{\epsilon} \right)^{o\left(1\right)} ,
\end{align}
where the right-hand side is a simplified expression if $\Omega \propto \eta$ (the thermodynamic limit).
The space complexity is ${\cal O}(\eta \log N\,\mathrm{polylog}(1/\epsilon))$ arising from the registers used to store the multipole information.
The results developed here represent roughly a speedup by ${\cal O}(\eta)$ over the most comparable prior methods. 

For a detailed comparison to prior work, see Table \ref{tab:comparison}.
For high filling fractions where $\eta$ and $N$ are similar sized, this reduces to roughly the same complexity as the best prior second-quantised plane wave algorithms developed by Su \emph{et al.}~\cite{Su2020} and Low \emph{et al.}~\cite{Low2018}.
In practice, it is expected that $N$ is significantly larger than $\eta$.
The algorithm with the best scaling for $N\gg\eta$ is the first-quantised interaction-picture method \cite{BabbushContinuum}, but that has much larger scaling with $\eta$.

To compare these complexities more precisely, we also take into account the volume $\Omega$.
The commonly considered case is the thermodynamic limit, where $\Omega$ is increased proportional to $\eta$.
In this limit, the result using the multipole approach has better complexity than that of the interaction picture method in Ref.~\cite{BabbushContinuum} given $N < \eta^7$.
(If $\Omega$ were constant, then our method would be best for $N < \eta^6$.)
As a result, the result presented here has the lowest known complexity for any approach whenever $N < \eta^7$, which is the result quoted in the abstract.
See Box \ref{box:summary} below for a summary of the method and resulting complexity.

\begin{table*}[t]
\begin{tabular}{|c|c|c|c|}
\hline
Year
& Reference
& Primary innovation
& Toffoli/T complexity\\
\hline\hline
2017
& Babbush \emph{et al}.~\cite{BabbushLow}
& Using plane waves with Trotter
& $\widetilde{\cal O}(\eta^2 N^{17/6} \sqrt{1 + \eta \Omega^{1/3}/N^{1/3}} /(\Omega^{5/6} \epsilon^{3/2}))$\\
2017
& Babbush \emph{et al}.~\cite{BabbushLow}
& Using plane waves with LCU
& $\widetilde{\cal O}((N^4/\Omega^{1/3} + N^{11/3}/\Omega^{2/3} )/\epsilon)$\\
2018
& Babbush \emph{et al}.~\cite{BGBWMPFN18}
& Linear scaling quantum walks
& $\widetilde{\cal O}((N^{10/3}/\Omega^{1/3} + N^{8/3}/\Omega^{2/3})/\epsilon)$\\
2018
& Low \emph{et al}.~\cite{Low2018}
& Interaction picture with second quantisation
& $\widetilde{\cal O}(N^{8/3} / (\Omega^{2/3} \epsilon))$\\
2018
& Babbush \emph{et al}.~\cite{BabbushContinuum}
& First quantised qubitisation
& $\widetilde{\cal O}((\eta^{3} (N/\Omega)^{1/3} + \eta^{2} (N/\Omega)^{2/3} ) / \epsilon)$\\
2018
& Babbush \emph{et al}.~\cite{BabbushContinuum}
& Interaction picture with first quantisation
& $\widetilde{\cal O}(\eta^{3} (N/\Omega)^{1/3} / \epsilon )$\\
2019
& Kivlichan \emph{et al}.~\cite{Kivlichan2020improvedfault}
& Better Trotter steps
& $\widetilde{\cal O}(N^{3} / (\Omega^{2/3} \epsilon^{3/2}))$\\
2019
& Childs \emph{et al}.~\cite{TrotComm}
& Tighter Trotter bounds
& $\cO(N^{7/3 + o(1)} / (\Omega^{1/3} \epsilon^{1 + o(1)})$\\
2021
& Su \emph{et al}.~\cite{Su2020}
& Tighter Trotter bounds for plane waves
& $N(\eta (N/\Omega)^{1/3}+ (N/\Omega)^{2/3})N^{o(1)} /\epsilon^{1 + o(1)}$\\
2023
& Low \emph{et al}.~\cite{LowPRX2022}
& Tighter Trotter in real space
& $N (\eta^{2/3}(N/\Omega)^{1/3} + (N/\Omega)^{2/3} ) (\eta N)^{o(1)} / \epsilon^{1+o(1)}$ \\
2024
& Rubin \emph{et al}.~\cite{RubinPNAS24}
& Tighter Trotter bounds in first quantisation
& $\eta^2 (\eta^{2/3}(N/\Omega)^{1/3} + (N/\Omega)^{2/3} ) (\eta N)^{o(1)} / \epsilon^{1+o(1)}$ \\
2025
& Stetina \& Wiebe~\cite{Stetina2025}
& Gauss’ law as a constraint
& $\eta^{2}(N/\Omega)^{2/3} (\eta N)^{o(1)} / \epsilon^{1 + o(1)}$ \\
2025
& \textbf{This work}
& quantum fast multipole
& $\eta (\eta^{2/3}(N/\Omega)^{1/3} + (N/\Omega)^{2/3} ) (\eta N)^{o(1)} / \epsilon^{1+o(1)}$ \\
\hline
\end{tabular}
\caption{\label{tab:comparison} Best quantum algorithms for phase estimating chemistry in a plane wave or real space basis.
For phase estimation the time $t$ is replaced with $1/\epsilon$ in the complexity.
$N$ is number of basis functions, $\eta < N$ is number of electrons, $\Omega$ is the computational cell volume, and $\epsilon$ is target precision.
}
\end{table*}

\begin{summarybox}[box:summary]{Summary of methods and results}
\begin{enumerate}
    \item Real-space first quantisation is used with $N$ grid points for $\eta$ electrons.
    \item Evolution for time $t$ with allowable error $\epsilon$ is simulated using high-order product formulae with the number of steps given in Eq.~\eqref{eq:trotter_number}.
\item The exponentials of the kinetic energy and potential are simulated by calculating these energies in the computational basis and applying phase factors.
\item The potential energy is calculated using a quantum form of the FMM in either one of two methods.
\begin{itemize}
    \item A nonadaptive method with a fixed grid (see Section \ref{sec:even}).
    The principle is that the electrons are moved into registers corresponding to the boxes for the FMM via an approach based on a quantum sort (see Algorithm \ref{alg:sortreg}).
    This method can only act on the subspace where the electrons are evenly distributed.
    \item An adaptive method that can account for an arbitrary distribution of particles (see Section \ref{sec:adaptive}).
    This works via three main principles.
    \begin{enumerate}
        \item The multipole information for each box is associated with a register for an electron, and is moved as the registers are sorted.
        \item The information from neighbouring boxes is obtained by copying it along a sorted list of electrons, according to Algorithm \ref{alg:copy}.
        \item The multipole information from all boxes in the interaction list is accessed by sorting the electrons in multiple Morton orderings as illustrated in Fig.~\ref{fig:Morton}.
    \end{enumerate}
\end{itemize}
\end{enumerate}
The order needed for the multipole method is $\trunc=\cO(\log(1/\epsilon)$, which results in the complexity for the quantum FMM for each step of the product formula being (see Appendix \ref{app:additional_fmm} and Appendix \ref{app:M2P})
\begin{equation}
    \cO(\eta\log N \log^4(1/\epsilon)).
\end{equation}
As this is linear in $\eta$, it is an asymptotic speedup over direct summation.
The qubit usage scales as
\begin{equation}
    \cO(\eta\log N \log^3(1/\epsilon)),
\end{equation}
and is primarily due to tree storage for the FMM.
\end{summarybox}

\section{Fast Multipole Method preliminaries}
\label{sec:fmm}
Here, we review the classical FMM for a system consisting of $\eta$ point particles. For the sake of clarity, we describe the non-adaptive case that is suitable for systems where the particles are roughly uniformly distributed within a cubic simulation cell.
Note that our quantum implementation does not rely on uniform distribution; in Section \ref{sec:adaptive} we describe the full adaptive algorithm.
Let $q_i$ and $\vb{r}_i$ ($i \in [\eta] \coloneqq \{1,\dots, \eta\}$) be the charge and position of the $i$th particle, respectively. Our objective is to compute the total potential energy given by 
\begin{equation}
    V = \frac 12 \sum_{i\in[\eta]}\sum_{\substack{j\in [\eta]\\j\neq i}}q_i \Phi(\vb{r}_i,\vb{r}_j)q_j = \sum_{i \in [\eta]} q_i V_i \, , \label{eq:total_potential_defn}
\end{equation}
for some kernel function $\Phi$, where $V_i = \frac 12 \sum_{j \in [\eta] \setminus \{i\}} \Phi(\vb{r}_i,\vb{r}_j)q_j$. In the case where all particles are electrons and we are interested in the Coulomb potential in 3D, we have $q_i = -1$ for all $i$ and $\Phi(\vb{r}_i,\vb{r}_j) = 1/\|\vb{r}_i -\vb{r}_j\|$ (using atomic units). Computing $V$ exactly involves summing over all pairs of particles, which requires $\cO(\eta^2)$ operations.

Instead, the FMM separates the potential experienced by each particle $V_i$ into near-field and far-field contributions,
\begin{equation}
    V = \sum_{i \in [\eta]} q_i V_i \ = \sum_{i \in [\eta]} q_i (V_{i,\mathrm{near}}+V_{i,\mathrm{far}}(\epsilon)). \label{eq:total_potential_defn_fmm}
\end{equation}
Here, $V_{i,\near}$ is evaluated exactly via summation over a number of geometrically close particles. The number of particles that are close to any given particle is bounded above by a constant.
In contrast, $V_{i,\far}$ is evaluated approximately via a Taylor expansion that rapidly converges to within a target error $\epsilon$ thanks to the low off-diagonal rank of $\Phi(\vb{r}_i,\vb{r}_j)$. The potential $V$ is evaluated in $\cO(\eta \,\mathrm{polylog}(1/\epsilon))$ operations using an appropriate choice of metric to distinguish near- and far-field contributions. This metric is defined in terms of a hierarchical tree-like division of the simulation cell, and the efficient evaluation of $V$ is typically framed in terms of the upward and downward traversal of this tree.

\subsection{Defining the tree}

An $L$-level hierarchical decomposition of the simulation cell defines a metric that determines whether the interaction between any pair of particles is accounted for in $V_{i,\near}$ or $V_{i,\far}$.
An example decomposition is illustrated in Fig.~\ref{fig:Interaction}.
This decomposition naturally maps onto a tree that is organised so that the root node (1st level) encompasses the entire simulation cell, the leaf nodes ($L$th level) partition the cell into $n_b$ boxes, and nodes at intermediate levels partition the cell into boxes with double the edge length, up to the root.
In $d$ dimensions, there are $2^{d(\ell-1)}$ boxes at the $\ell$th level of this tree.
While we will sometimes use a $d=2$ quadtree for illustrative purposes in figures, we are generally concerned with a $d=3$ octree in the evaluation of $V$ for electronic structure problems.

The governing principle in FMM is that the interaction between particles in boxes that are well separated can be efficiently evaluated by
\begin{enumerate}[label=(\arabic*)]
    \item aggregating the distribution of ``source'' particles (w.l.o.g., those indexed in the sum over $j$ in Eq.~\eqref{eq:total_potential_defn}) into multipole expansions at all levels of the tree, 
    \item translating the multipole expansions defining the attendant potential contribution into local (Taylor) expansions about boxes containing ``observer'' particles (w.l.o.g., those indexed in the sum over $i$ in Eq.~\eqref{eq:total_potential_defn}), and
    \item disaggregating the local expansions into the potential experienced by any observer particle (i.e., $V_{i,\far}$).
\end{enumerate}
The remaining contributions to the potential experienced by any observer particle (i.e., $V_{i,\near}$) are due to particles in boxes that are \emph{not} well separated. These interactions are evaluated exactly.

While the near-field interactions are strictly accounted for among leaf-level boxes, evaluating the far-field interactions takes advantage of the fact that boxes separated by increasingly larger distances are more efficiently accounted for at higher levels of the tree. Performing these higher-level calculations requires defining the interaction list for any given box $b$ -- the set of boxes that are not too close to $b$, but also not too far away. In particular, two boxes are in each other's interaction list when they are not nearest neighbours but their parent boxes (at the next level up) are.
This is illustrated in Figure \ref{fig:Interaction}.

To make this concept more rigorous, we define a few concepts. At a given level $\ell$, let $B_\ell$ be the set of all boxes at this level.
The size of $B_\ell$ in 3D is $|B_\ell| = 8^{\ell - 1}$. For a box $b \in B_\ell$, we denote the set of nearest neighbours to $b$ as $\nearn(b)$.
Here, we define the nearest neighbours of $b$ to be all boxes that share an edge or vertex with $b$. For $\ell \neq 1$, the size of $\nearn(b)$ in 3D is $|\nearn(b)| \leq 26$.
The parent box $\parent(b) \in B_{\ell-1}$ is the box one level up from $b$ such that $b$ is fully enclosed by the volume of $\parent(b)$.
We also define the child boxes of $b$ to be boxes at level $\ell+1$ that are enclosed by the volume of $b$.
The set of child boxes of $b$ is denoted $\child(b)$.
Note that $b \in \child(\parent(b))$ is always true. Additionally, let $p(b)$ be a set of integers that index the particles within a given box $b$. The interaction list of $b$ is given by 
\begin{equation}
    \inter(b) = \child(\nearn(\parent(b))) - \nearn(b) - b,
\end{equation}
where the final term is needed so that $b$ itself is not included in the interaction list. Finally, let $\mathbf{c}_b$ be a vector representing the position of the centre of box $b.$ The notation introduced here is summarised in Table \ref{tab:notation}.  

\begin{table}[t] \label{tab:notation}
  \centering
  \begin{tabular}{@{}| l | l |@{}}
    \hline
    $B_{\ell}$         & set of boxes on level $\ell$ \\
    $\nearn(b)$            & nearest neighbours of box $b$ \\
    $P(b)$             & parent box of $b$ \\
    $C(b)$             & set of child boxes of $b$ \\
    $I(b)$             & interaction list for $b$ \\
    $p(b)$             & indices of particles contained in $b$ \\
    $\mathbf{c}_{b}$   & position of the centre of box $b$ \\
    \hline
  \end{tabular}
\caption{Summary of notation for boxes.}
\end{table}

\begin{figure}
    \centering
    \includegraphics[width=0.4\linewidth]{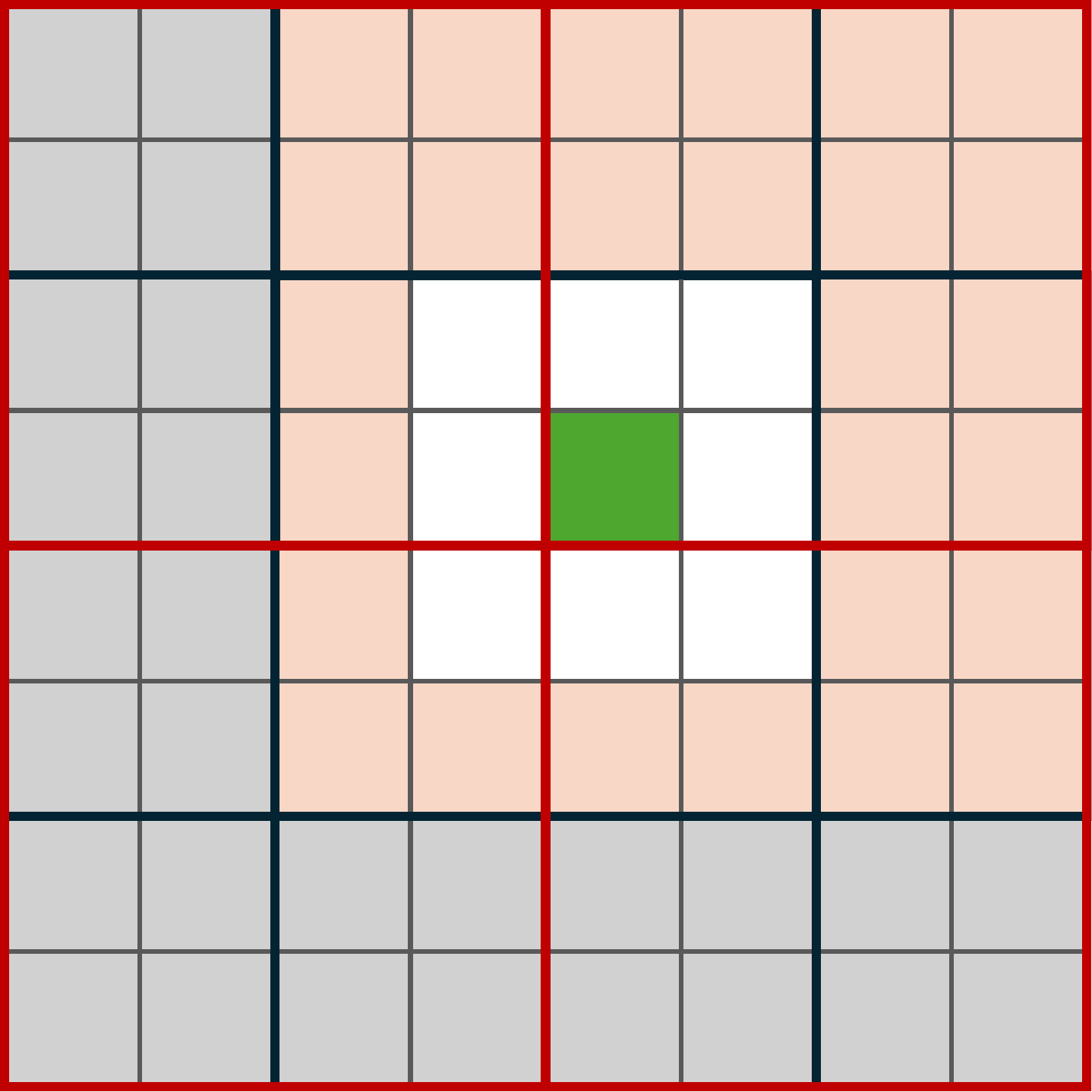}
    \caption{Illustration of hierarchical division of a simulation cell, as well as the interaction list and neighbours for an exemplary (green) box in the $\ell=4$th level of a tree. Here, the level-1 box is the complete region. The four level-2 boxes are outlined by the red lines. The 16 level-3 boxes are shown by the thick black lines and the level-2 partition. The 64 level-4 boxes are shown by the thin black lines and the level-3 partition. The box of interest $b$ is represented by the green square. The surrounding white squares are neighbours of $b$, and the boxes in $b$'s interaction list are shown in orange. The grey squares are neither nearest neighbours, nor in the interaction list. A two-dimensional simulation cell is presented for simplicity of presentation.}
    \label{fig:Interaction}
\end{figure}

The interaction list of any box has at most $6^3 - 3^3 = \cO(1)$ boxes in 3D. To see this, note that along any dimension the width of $\child(\nearn(\parent(b)))$ is $6$ boxes wide. In 3D, this yields $6^3$ total boxes. Then, notice that the width of $\nearn(b)$, including the centre box $b$, along any dimension is $3$, for a total of $3^3$ in 3D. Finally, the total number of boxes in the interaction list is given by the difference, $6^3 - 3^3=189$. There can be fewer boxes near the boundaries of the region.

\subsection{Traversing the tree}
Here, we consider an illustrative zeroth-order (monopole only) implementation, and provide additional details required to achieve arbitrary $\epsilon$ using higher orders of the multipole expansion in Appendix~\ref{app:additional_fmm}.
In Algorithm \ref{alg:FMM}, we provide pseudocode that implements a simplified FMM that returns an approximation to $V$.
This version of the algorithm only calculates the total far-field contribution to the potential between boxes and does not disaggregate those interactions to the local potential experienced by each particle.
Later, we present a version of the algorithm that calculates these interactions via disaggregation (vide infra Algorithm~\ref{alg:FMMagg}).
The algorithm consists of two passes through the hierarchical tree in which every box at levels $\ell\ge 3$ is iterated through. In the upward pass, the objective is to compute the charge $Q_b$ for each box so that it can be used in the downward pass to compute an approximation of $V$. The upward pass starts at level $\ell =\Levs$ and works up to $\ell = 3$.
At $\ell = L$ the total charge $Q_b$ of the particles in each box $b$ is computed. Then, at $\ell =\Levs-1$, the total charge of each box is computed by adding the charges of the box's $8$ child boxes computed previously. We continue in this way until $\ell=3$ is reached. If there are at most $c = \cO(1)$ particles per box at level $\Levs$, computing $Q_b$ for a given box requires $\cO(1)$ operations. The total number of operations is then $\cO(n_b)$ (omitting the $\epsilon$-dependence).

The objective of the downward pass is to compute an approximation to $V$ using the monopoles computed in the upward pass.
This is done starting at level $3$ and working up to level $\Levs$.
The first two levels are skipped because none of the boxes in $B_1$ and $B_2$ are sufficiently separated from the remaining boxes for the multipole approximation to be valid.
This is the reason the box charges are only calculated up to level $3$ in the upward pass.
For all other levels, the interactions between a given box $b$ and the boxes in its interaction list are computed using the multipole approximation.
At the leaf-level, the interactions between particles occupying the same box are computed exactly using the Coulomb potential.
Additionally, the interactions between particles in a given leaf-level box and particles in the nearest-neighbour boxes are computed exactly. 
All of the interactions computed throughout the downward pass are aggregated in a variable $V$ that is returned after the algorithm completes.

In FMM, rather than simply summing the potential energy into $V$, the downward pass translates the multipole expansions into local expansions and disaggregates those local expansions into potentials at each leaf-level box.
In Algorithm \ref{alg:FMMagg} we present a version of the FMM using this disaggregation, though still using the monopole approximation for simplicity.
It uses, for example, $V_b \leftarrow V_b + \Phi(\vb{c}_a,\vb{c}_b)Q_a$ rather than $V \leftarrow V + Q_b\Phi(\vb{c}_a,\vb{c}_b)Q_a$, with the multiplication by the box charge $Q_b$ omitted to give the box potential rather than the potential energy.
The potentials for parent boxes are copied to child boxes, and eventually to the potentials for particles.
These are then multiplied by the particle charges to determine the potential \emph{energy}.

In the monopole-only case described here, these two approaches are completely equivalent, but
in the higher-order calculation, the disaggregation simplifies the calculation.
Determining the potential energy between boxes based on the multipole (as in the direct translation of Algorithm \ref{alg:FMM}) would introduce significant additional complexity.
The complexity is reduced by using the multipole information for the boxes in the interaction list to determine the potential at box $b$.
A detailed explanation of the higher-order method is given in Appendix \ref{app:additional_fmm}.

Because each interaction list contains $\cO(1)$ boxes, and each level-$\Levs$ box has $\cO(1)$ neighbours (each with $\cO(1)$ particles), the downward pass also requires $\cO(n_b)$ operations in total. If the particles are roughly uniformly distributed, then $\Levs$ can be chosen such that $n_b = \cO(\eta)$, which enables the calculation of $V$ with complexity $\cO(\eta)$.
The complexity has further polynomial factors of $\log(1/\epsilon)$ for the full multipole calculation.
For more general distributions, an adaptive subdivision of boxes can be used to limit the number of boxes that need be considered, providing similar complexity.

\begin{algorithm}
    \caption{Classical FMM (monopole only, no disaggregation)}\label{alg:FMM}
    Upward pass: compute charges\\
    1. \For{$b\in B_\Levs$}{
    $Q_b \leftarrow \sum_{j \in p(b)} q_j$
    }
    \BlankLine
    2. \For{$\ell \leftarrow\Levs-1$ \KwTo $3$}{
    \For{$b\in B_\ell$}{
    $Q_b \leftarrow \sum_{\text{$a\in\child(b)$}} Q_{a}$
    }
    }
    \BlankLine
    Downward pass: evaluate potential\\
    initialise $V\leftarrow0$\\
    Compute far-field contribution using box-box potentials\\
    3. \For{$\ell \leftarrow 3$  \KwTo \Levs}{
    \For{$b\in B_\ell$}{
    \For{$a\in\inter(b), a > b$}{
    $V \leftarrow V + Q_b\Phi(\vb{c}_a,\vb{c}_b)Q_a$
    }
    }
    }
    \BlankLine
    Compute near-field contribution using particle-particle potentials\\
    4. \For{$b \in B_{L}$}{
        \For{$i \in p(b)$}{
            \For{$n \in \nearn(b), n>b$}{
                \For{$j \in p(n) $}{
                Compute Coulomb interactions between particles within $b$ and $n$\\
                $V \leftarrow V +q_i\Phi(\mathbf{r}_i,\mathbf{r}_j)q_j$
                }
            }
            \For{$j \in p(b), \,j>i$}{
                Compute Coulomb interactions between particles within $b$\\
                $V \leftarrow V + q_i\Phi(\mathbf{r}_i,\mathbf{r}_j)q_j$
            }
        }
    }
    Return $V$
\end{algorithm}

\begin{algorithm}
    \caption{Classical FMM (monopole only, with disaggregation)}\label{alg:FMMagg}
    Upward pass: compute charges\\
    1. \For{$b\in B_\Levs$}{
    $Q_b \leftarrow \sum_{j \in p(b)} q_j$
    }
    \BlankLine
    2. \For{$\ell  \leftarrow\Levs-1$ \KwTo $3$}{
    \For{$b\in B_\ell$}{
    $Q_b \leftarrow \sum_{\text{$a\in\child(b)$}} Q_{a}$
    }
    }
    \BlankLine
    Downward pass: evaluate potential\\
    initialise $V_b\leftarrow0$\\
    Compute far-field contribution \\
    3. \For{$\ell \leftarrow 3$  \KwTo \Levs}{
    \For{$b\in B_\ell$}{
    $V_b \leftarrow V_{\parent(b)}$\\
    \For{$a\in\inter(b), a > b$}{
    $V_b \leftarrow V_b + \Phi(\vb{c}_a,\vb{c}_b)Q_a$
    }
    }
    }
    \BlankLine
    Compute near-field contribution \\
    initialise $V\leftarrow0$\\
    4. \For{$b \in B_{L}$}{
        \For{$i \in p(b)$}{
            $V_i \leftarrow V_b$\\
            \For{$n \in \nearn(b), n>b$}{
                \For{$j \in p(n) $}{
                Compute Coulomb interactions between particles within $b$ and $n$\\
                $V_i \leftarrow V_i +\Phi(\mathbf{r}_i,\mathbf{r}_j)q_j$
                }
            }
            \For{$j \in p(b), \,j>i$}{
                Compute Coulomb interactions between particles within $b$\\
                $V_i \leftarrow V_i + \Phi(\mathbf{r}_i,\mathbf{r}_j)q_j$
            }
        $V \leftarrow V + q_i V_i$
        }
    }
    Return $V$
\end{algorithm}

\section{Quantum FMM for evenly distributed particles} \label{sec:even}
The difficulty in constructing a quantum implementation of the FMM is that we need to place electrons into boxes according to their positions, which are stored in quantum registers. This means that we need to access data according to a value stored in a quantum register. For each electron, we need to run through $\cO(\eta)$ boxes to determine the box that a given electron belongs to. Because there are $\eta$ electrons this approach results in $\cO(\eta^2)$ cost, which we are trying to avoid. We first illustrate the principles used to avoid $\cO(\eta^2)$ scaling for the simpler case where the particles are assumed to be evenly distributed. In Section \ref{sec:adaptive}, we modify our approach to allow for adaptive grids which lets us treat the general case where the evenly distributed assumption is lifted. We note that within the BO approximation the positions of the nuclei are given by classical registers and there is no difficulty moving data for the nuclei.
For non-BO simulations the registers storing the positions of the nuclei would need to be moved into boxes as well. However, this can be achieved with a similar approach as for the electron registers.

\subsection{Quantum registers for FMM data storage}
Our implementation of the quantum FMM algorithm uses several quantum registers to store information used throughout the algorithm. For a given number of electrons $\eta$, let $n_b = \cO(\eta)$ so that the maximum number of electrons per box at the leaf level is $c$ and the overall complexity is $\widetilde{\cO}(\eta)$.
For each box $b$, we have the registers:
\begin{itemize}
    \item $\ket{Q_b}$ -- the total charge (monopole) of $b$,
    \item $\ket{V_b}$ -- the potential at box $b$.
\end{itemize}
In the higher-order FMM, $Q_b$ and $V_b$ will be replaced by the multipole and local expansion coefficients, respectively.
In the monopole case, registers for $V_b$ are not strictly needed because we could perform the algorithm in the simplified form in Algorithm \ref{alg:FMM}.
Quantum registers for $V_i$ (the potentials at each particle) are not needed, because any particle's far-field contribution to the potential energy can be directly added into $V$ without storing each $V_i$.

In addition to these registers, boxes at the leaf level have the additional registers:
\begin{itemize}
    \item $\ket{\mathbf{r}_i}, \,\,i\in p(b)$ -- $c$ registers storing the position of each particle,
    \item $\ket{q_i}, \,\, i \in p(b)$ --  $c$ qubits which flag whether each of the $|\mathbf{r}_i\rangle$ has been populated with information for an electron.
\end{itemize}
That is, at the leaf level each of the $n_b$ boxes is given $c$ registers, which are used to store the positions of the electrons that occupy the box. Additionally, each of these $n_b c$ position registers has an associated flag register that is set to $|1\rangle$ if the position register is occupied, and is $|0\rangle$ otherwise.
More generally, we can allow $|q_i\rangle$ to encode the charge for a nucleus, if we include nuclei in the FMM.
We group together these position and flag registers as what we will call a ``particle'' register, and will consider operations moving data that move both the position and flag registers together.

Although we have described the particle registers with respect to the leaf-level boxes, they are used for boxes at all levels of the tree.
In particular, the registers of box $b$ at the leaf level are also used for the boxes in $P(b)$.
Hence, each box in $\ell = L-1$, has $2^dc$ particle registers for dimension $d$.
For the description of the quantum FMM it is convenient to regard these registers to be arranged in the same way as the spatial locations.

Each position $\vb{r}$ is in a space $G$ of $N$ grid points.
We encode each $\vb{r} \in G$ as $\ket{\vb{r}} = \ket{x}\ket{y}\ket{z}$ where each $\ket{\alpha}$ ($\alpha \in \{x,y,z\}$) is a binary representation of the spatial coordinate $\alpha$.
This uses $3 \lceil\log(N^{1/3})\rceil$ qubits (where we adopt the usual quantum information convention that logs are to base 2).
Position registers that do not encode an electron are set to the all-zero state.
When interleaving the bits for the Morton ordering, we encode the position as $\ket{\vb{r}}=\bigotimes_{k=1}^{n_{p}}\ket{x_{k}}\ket{y_{k}}\ket{z_{k}}$ for $k\in\{1,2,..., L,...,n_{p}\}$, where $n_{p}=\lceil\log(N^{1/3})\rceil$.

We store the position by natural numbers starting from zero, rather than numbers in the range $[-(N^{1/3}-1)/2,(N^{1/3}-1)/2]$ as indicated by $G$ in Eq.~\eqref{eq:vardefs}.
This is because the natural numbers are more convenient for defining the tree structure for the FMM.
This encoding is for the position, but
for the kinetic part of the Hamiltonian we need momenta that are symmetric about zero.
In the implementation of the Hamiltonian this can be accounted for in how the Fourier transform is implemented.

With this encoding of the position, we can easily determine which box an electron should be placed in based on the encoding of its position.
For example, at the highest level the triple of the first (most significant) bits $(x_{1},y_{1},z_{1})$ of the $x$, $y$ and $z$ coordinates of the electrons naturally partitions the box into 8 equal sections.
Each of these 8 boxes is further subdivided into 8 more boxes by the triple of next most significant bits $(x_{2},y_{2},z_{2})$.
This can then be continued up to the triple $(x_{L},y_{L},z_{L})$.
As a result, we do not need additional resources to determine which box a given electron should be moved to.

\subsection{Moving electron data to the correct  registers}
Before computing the potential using the FMM, we need to move the data for each of the $\eta$ electrons into the position registers for the appropriate leaf-level boxes.
This is achieved in a recursive way, where if the electron data is in the correct boxes at level $\ell$, then within each box the electron data is moved into the correct child boxes by a procedure of sorting the data and swapping it into the appropriate child box.
We first describe how to move electron position data into the correct boxes for the 1D case, then generalise to 3D. 
Figure \ref{fig:algo2} illustrates this procedure for the 1D case.

First, the system registers storing the electron positions are used as the first $\eta$ of the particle registers, with the flag registers initialised to $\ket{1}$.
The remaining $n_b c -\eta$ particle registers are initialised to the all-zero state, so $|\mathbf{0} \rangle = |0\dots 0\rangle$ and $|0\rangle$ for the position and flag registers, respectively.
Next, the particle registers are sorted according to the position.
The complexity of the sort is $\cO(\eta\log\eta)$ comparisons and controlled swaps, which yields a gate complexity of $\cO(\eta\log\eta\log N)$.
The convention for each sort is that the flag registers are used to ensure that the populated registers are moved to the left.
Note that we are sorting the particle registers, so the flag registers are moved with the position registers.

After sorting, we divide the complete simulation cell (level-1 box) into two equally sized level-2 boxes. Then, the data in the registers is swapped so that the electron positions are placed in the appropriate level-2 boxes. To achieve this task, we run through all $n_b c /2$ registers in the left box, or $\min(n_b c,\eta)$ locations accounting for the fact that not all of these registers need contain data for electrons. For each position register, we check the most significant qubit of $\vb{r}$. If the qubit is in the state $|0\rangle$, the electron belongs to the left box and does not need to be swapped. If the qubit is in the state $\ket{1}$, the electron belongs to the right box, and we swap it. If the electron was in register $m$, we swap the data to register $m + n_b c/2$, which is register $m$ of the right box.
As before, we are swapping the particle register, with both the position and flag registers.
The complexity of the controlled swaps is $\cO(n_b c\log N)=\cO(\eta\log N)$.

When the data of a register is swapped to the right, it is important that the register it is swapped with does not contain an electron position (so the flag register is in the $\ket{0}$ state). 
This can be proven for our procedure in the following way.
If there are initially no electrons in the registers for the right box, then it is trivially true.
If there are electrons initially in the registers for the right box,
then all data in registers $m$ to $n_b c/2-1$ needs to be swapped to the right (counting from 0).
This is because the registers are initially sorted.
If it were the case that there was an electron already stored in register $m+n_b c/2$, then it would mean that there were more than  $n_b c/2$ electrons that should be in the right box.
This is more electrons than can be stored in the registers, violating the even distribution assumption in this section.
Hence, this procedure never swaps data for an electron from the right box into the left box.

After swapping the data to the correct registers, we perform a sort on the registers of the right box with complexity $\cO(\eta \log\eta\log N)$. The left box is already guaranteed to be sorted because the order was not changed. Each of the two level-2 boxes is then subdivided into two level-3 boxes.
We then perform the controlled swaps within each level-2 box to ensure the electrons are within the correct level-3 boxes.
We proceed in this way until we swap the data into the correct level-$\Levs$ boxes. There are $\log n_b$ levels, each with leading-order complexity for the sort $\cO(\eta \log\eta\log N)$, as well as $\cO(\eta\log N)$ complexity for controlled swaps, for a total complexity
\begin{equation}
    \cO(\eta \log^2\eta \log N) \, .
\end{equation}

In three (or any number of) dimensions, we can apply the result for one dimension by choosing an ordering for the boxes.
For example, we can use a Morton ordering where the bits of the $x$, $y$, and $z$ directions are interleaved. Alternatively, we can give the bits for $x$ first, then $y$, then $z$. For either ordering, the same procedure can be used as in the 1D case to move the electrons to the correct boxes. The only difference is that the interpretation of the subdivision of the region into two corresponding to child boxes no longer holds.

\begin{figure}
    \centering
\includegraphics[width=\textwidth]{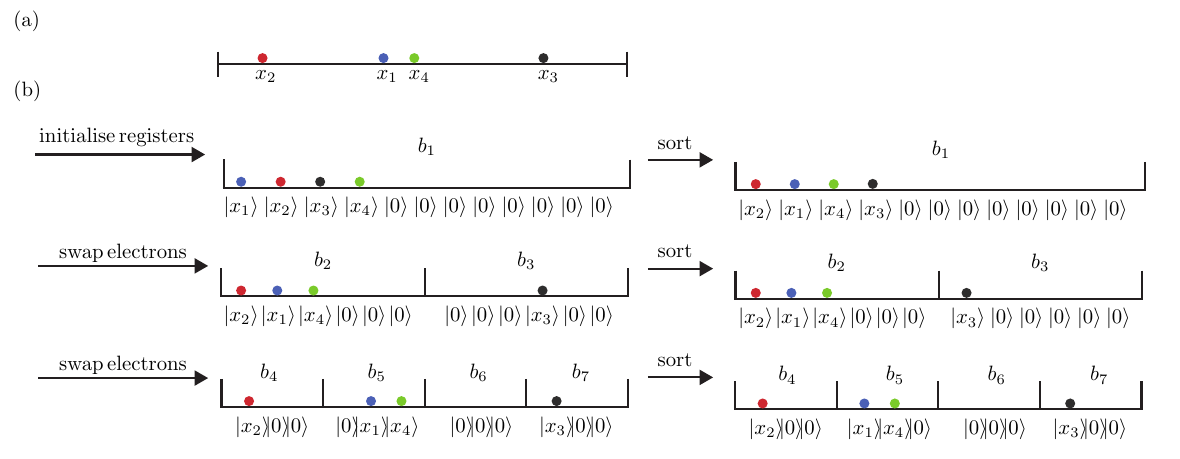}
    \caption{1D demonstration of Algorithm \ref{alg:sortreg}. Here $\eta=4$, $L=3$, $c=3$, and $n_b = 4$. (a) The positions of the electrons are shown on the real space grid, where, $x_i$ is the position of the $i$th electron. (b) Snapshots of the boxes and position registers throughout Algorithm \ref{alg:sortreg}. The label of a given box is shown above and the corresponding registers are shown below. The flag registers are not depicted. Each row iterates through the outer-most loop of Algorithm \ref{alg:sortreg} for $\ell=1,2,3$, so the successive rows show the division of the region into boxes. For a given $\ell$, the iteration starts by swapping position data into the correct registers. For $\ell=1$, this step is replaced by an initialisation step. Next, the positions are sorted within each box. This procedure is repeated until $\ell=L$ is reached.}
    \label{fig:algo2}
\end{figure}

The complete procedure is described in Algorithm \ref{alg:sortreg}.
In this algorithm, it is assumed that there are initially $8^L c$ position and flag registers. The $\eta$ electron position registers are input as part of these $8^L c$ registers with the associated flag registers set to $|1\rangle$. All other position registers are set to $|\mathbf{0}\rangle$, and flag registers to $|0\rangle$. All sorts used within this algorithm are quantum sorts, and move the flag registers together with the position registers, as well as moving occupied registers to the left.

\begin{algorithm}
    \caption{Sorting electron registers}\label{alg:sortreg}
    \For{$l \gets 0$ \KwTo $3(\Levs-1)-1$}{
	\For{$b\gets 0$ \KwTo $2^l-1$}{
		$\boxwid \gets c\times 2^{3(L-1)-l}$\\
		$\firstbox \gets b\boxwid$\\
		Sort registers $\firstbox$ to $\firstbox+\boxwid-1$.\\
		\For{$j\gets \firstbox$ \KwTo $\firstbox+\boxwid-1$}{
			\If{$p_{j} \ge (\firstbox+\boxwid/2)/c \times N/2^{3(L-1)}$}{
				Swap register $j$ with register $j+\boxwid/2$.
			}
		}
	}
}
\end{algorithm}

Here, the quantity $l$ counts the number of times the region has been divided by 2, rather than the level, and  $b$ is an integer corresponding to which region is being considered, rather than a box of the octree.
The index $j$ counts through the registers in the sorted sequence, and $p_j$ is the value in that register interpreted as an integer, which may, for example, correspond to interleaved bits of $x$, $y$, and $z$ for Morton ordering.
The test $p_{j} \ge (\firstbox+\boxwid/2)/c \times N/2^{3(L-1)}$ is used to check if the register is in the wrong half of the region.

\subsection{Calculation of the potential given correct data locations}

Given that the electrons are placed in the correct registers, it is straightforward to calculate the potential via the FMM, directly translating the steps in Algorithm \ref{alg:FMM} into quantum gates.
The algorithm as described above can be directly translated into a quantum algorithm using standard methods of performing coherent quantum arithmetic.
The parts of the calculation that directly use the information for the electrons in the boxes are the first step of loop 1 (which sums the charges in a box for $\ell=\Levs$), and loop 4 (which sums the potential over electrons in neighbouring boxes).

For these calculations we iterate over all $c$ registers in each box, but only add the contribution if the corresponding flag qubit stores $\ket{1}$.
In particular, for $Q_b \leftarrow \sum_{j \in p(b)} q_j$, we sum the values stored in these flag registers which gives the total number of electrons in the box.
Then, for $V_i \leftarrow V_i + \Phi(\vb{r}_i,\vb{r}_j)q_j$, we need not directly store $V_i$, and instead add $q_i\Phi(\vb{r}_i,\vb{r}_j)q_j$ into $V$.
This means that we perform a Toffoli on the flag registers for electrons $i$ and $j$, and use the result to control addition of the potential term $\Phi(\vb{r}_i,\vb{r}_j)$.

The electron-nuclear potential for $U$ is easily accounted for by adding the nuclear charges for the boxes in loop 1 of Algorithm \ref{alg:FMM}, and for $V_i \leftarrow V_i +  \Phi(\vb{r}_i,\vb{r}_j)q_j$ in loop 4, summing over nuclear charges as well.
That is, $q_j$ also includes the nuclear charges, and we can use $q_i$ to control adding $\Phi(\vb{r}_i,\vb{r}_j)q_j$ into the total potential.
The sum over nuclear charges is easily performed because the positions are given classically, so it is known which nuclei are within each box without requiring any quantum data accesses.

The remaining steps involve computations on other registers associated with boxes, which are independent of how many particles are in each box.
These include summing charges from child boxes (loop 2), and summing contributions to the potential from boxes in the interaction list (loop 3).
These operations all involve additions, except loop 3 also requires a multiplication.
The value of $\Phi(\vb{c}_a,\vb{c}_b)$ used in loop 3 does not need to be computed using coherent arithmetic, because we classically iterate over the box numbers $a$ and $b$, and the values of $\Phi(\vb{c}_a,\vb{c}_b)$ can be precalculated classically.
The only case where potentials need be calculated via coherent arithmetic is for the interparticle potentials in loop 4.

Overall the complexity of these calculations is as follows.
\begin{itemize}
    \item The complexity of adding the numbers of electrons in the $n_b$ boxes is $\cO(n_b c)$.
    \item For adding the charges, at level $\ell$, there are $3(\Levs-\ell)+\lceil\log c\rceil$ qubits to perform arithmetic on, and $n_b/2^{3(\Levs-\ell-1)}$ additions to perform.
    Summing this complexity over $\ell=1$ to $\Levs-1$ gives total complexity $\cO(n_b \log c)$.
    \item At level $\ell$, $Q_a$ has $3(\Levs-\ell)+\lceil\log c\rceil$ qubits, so the complexity of multiplying by the potential with $\cO(\log(1/\epsilon))$ qubits is $\cO([3(\Levs-\ell)+\log c]\log(1/\epsilon))$.
    Summing this complexity over the boxes and levels gives total complexity $\cO(n_b \log c \log(1/\epsilon))$.
    \item For each leaf-level box there are $\cO(c^2)$ pairwise potentials to approximate, for a total of $\cO(n_b c^2)$
    for all boxes.
    The values of $\Phi(\vb{r}_{i},\vb{r}_{j})$ can be approximated by using a QROM for interpolation combined with Newton's method \cite{Haner2018}.
    In practical usage the number of iterations of Newton's method can be constant, but the number of iterations for error $\epsilon$ is $\cO(\log\log(1/\epsilon))$.
    Each step of Newton's method uses multiplication, with complexity of $\cO(\log^2(1/\epsilon))$, yielding a total complexity
    \begin{align}
        \cO\left(n_b c^2 \log^2(1/\epsilon)\,\log\log(1/\epsilon)\right).
    \end{align}
    \item Multiplying $V_i$ by $q_i$ for each electron and adding into $V$ corresponds to a controlled addition, so has total complexity $\cO(\eta\log(1/\epsilon))$.
\end{itemize}
Out of these complexities, the dominant complexity is that from summing the pairwise potentials between electrons, which for $n_b=\cO(\eta)$ and $c$ a constant is
\begin{equation}
    \cO(\eta \log^2(1/\epsilon)\,\log\log(1/\epsilon)) \, .
\end{equation}
Note that in this analysis we allow the error from each step in the calculation to be $\cO(\epsilon)$, which is why the number of bits for the arithmetic is taken to be $\cO(\log(1/\epsilon))$.
In practice the error in each step will need to be smaller to ensure the total error is no larger than $\epsilon$, but that has no effect on the complexity given in the form of an order scaling in $\log(1/\epsilon)$.
When including higher orders in the multipole expansion, the order should be taken to be $\cO(\log(1/\epsilon))$, the number of multipole moments scales as the square of the order, and each should be given to $\cO(\log(1/\epsilon))$ bits.
In addition, $V_b$ is replaced with a number of local expansion coefficients scaling as the square of the order, and there is additional arithmetic needed to translate this information to child boxes.
These features result in a higher power in the scaling with $\log(1/\epsilon)$; see Appendix \ref{app:additional_fmm} for details.

\section{Adaptive boxes for arbitrary distributions of particles}
\label{sec:adaptive}

In the case where particles are distributed in an uneven way, the boxes are chosen in an adaptive way.
That is, instead of subdividing the boxes down to a fixed level, subdivide boxes that contain multiple particles.
That enables boxes of a very small size, but no more than $\eta$ of these boxes need be retained because the others have no particles.
The difficulty is that the boxes need to be determined based on the numbers in registers, which would mean data accesses at locations according to these values in a simple translation of the classical method.
That would result in an overall complexity of at least $\eta^2$, eliminating the speedup.
In order to implement the procedure with a speedup we again need a way to perform the computation with data accesses at fixed locations.

\subsection{Calculating charge / multipole information}
\label{sec:calcmult}
First we consider how to calculate the charge and multipole information for the boxes, then consider the more difficult question of how to compute the potential information, which requires locating data for boxes in the interaction list.
The general principle we use is, for each particle include a set of data registers for every box that the particle is contained in.
In contrast to the non-adaptive case, we do not include a flag qubit or additional unoccupied particle registers.
We first explain the method in one dimension, then how to generalise it to multiple dimensions.
We also describe the method for the electron-electron potential $V$ first, then an amendment for the electron-nuclear potential.

The difficulty is primarily in finding the data locations for boxes in order to compute the results for multiple boxes.
The method we use to solve this is to copy data along the chain of  registers for electrons, so that data for one box is copied to a neighbouring box in order to perform the calculation.
In one dimension, the electron registers are first sorted into order.
Then we can have a sequence of electron registers that are part of one box at level $\ell$, followed by electron registers in a neighbouring box (or possibly another box).

In one dimension, we would calculate the charge for the box at level $\ell-1$ by adding the charges of the child boxes at level $\ell$.
Each box at level $\ell$ has a register containing that charge information for level $\ell$, but that information needs to be copied over to the registers for charges in the next child box in that level in order to compute the charge at level $\ell-1$.
Similarly, the multipole information needs to be copied over to compute the multipole information at level $\ell-1$.

In particular, let us denote the charge information for the box that electron register $j$ is in at level $\ell$ by $Q(j,\ell)$.
We explain the procedure for charge, and the procedure for multipole information is equivalent, except the operations to combine the multipole information are more complicated.
Let us also denote by $\boxno{j,\ell}$ the bits up to $\ell$ of the electron register, which is the number of the box at level $\ell$.
We take $\Levs$ to be the total number of bits for the positions of the electron registers, so all electrons are in different boxes at level $\Levs$.
The procedure is then as in Algorithm \ref{alg:calccharg}.
This is the quantum implementation of the upward pass of the FMM in Algorithm \ref{alg:FMMagg}.

\begin{algorithm}
    \caption{Calculating charge information}\label{alg:calccharg}
1. \For{$j\gets 0$ \KwTo $\eta-1$}
    {$Q(j,\Levs)\gets 1$ \qquad (corresponding to a single charge for this finest-level box)}
2. \For{$\ell\gets L-1$ \KwTo $3$ \KwBy $-1$}{
    (a) \For{$j\gets 0$ \KwTo $\eta-2$}{
        \If{$\boxno{j+1,\ell}=\boxno{j,\ell}$}{
            \uIf{$\boxno{j+1,\ell+1}\ne \boxno{j,\ell+1}$}{
                $Q(j+1,\ell) \gets  Q(j,\ell+1)$
            }
            \ElseIf{$\boxno{j,\ell+1} \mod 2 =1$}{
                $Q(j+1,\ell) \gets  Q(j,\ell)$
            }}}
    {(b) \For{$j\gets \eta-2$ \KwTo $0$ \KwBy $-1$}{
        \If{$\boxno{j,\ell} = \boxno{j+1,\ell}$}{
            \uIf{$\boxno{j,\ell+1}\ne \boxno{j+1,\ell+1}$}{
                $Q(j,\ell) \gets  Q(j+1,\ell+1)$
            }
            \ElseIf{$\boxno{j,\ell+1}\mod 2 =0$}{
                $Q(j,\ell) \gets  Q(j+1,\ell)$
            }}}}
    (c) \For{$j\gets 0$ \KwTo $\eta-1$}{
        $Q(j,\ell) \gets  Q(j,\ell)+ Q(j,\ell+1)$
}}
\end{algorithm}

This algorithm first sets the charge at the leaf level.
Then in 2(a) it runs forward through the sequence of registers, and checks if the neighbouring registers are for the same box at level $\ell$ but different boxes at level $\ell+1$.
If they are, then it copies the charge from the left level-$\ell+1$ box into the register for the right level-$\ell+1$ box, but in the register being used for level-$\ell$ information.
The other possibility is if the least significant bits of $\boxno{j+1,\ell+1}$ and $\boxno{j,\ell+1}$ are equal to $1$ (checked using mod $2$), in which case we have a pair of registers in the right level-$\ell+1$ box.
Then the appropriate information will be in $Q(j,\ell)$, and is just copied into $Q(j+1,\ell)$.

In this way, every time we hit a boundary between two level-$\ell+1$ boxes that need their charges summed, we copy the charge information from the left box into the sequence of registers in the right box.
If it happened that there were no electrons in the left box, then the $\boxno{j+1,\ell}=\boxno{j,\ell}$ test would fail, and we would just leave the $Q(j+1,\ell)$ register as zero.
Then 2(b) just does the same procedure running from right to left, copying information from the right boxes into the left boxes.
This ensures that at level $\ell$, $Q(j,\ell)$ contains the charge from the \emph{other} level-$\ell+1$ box.
Then 2(c) adds the charges, so $Q(j,\ell)$ contains the total charge for the level-$\ell$ box (adding the two charges of the child boxes at level $\ell+1$).
For the multipole information the procedure is identical, except the operation in 2(c) to combine multipole information is not just addition.

The same principle can be used in higher dimensions.
In 1D the division into boxes is simply repeated division of boxes into two.
In 2D we can take the complete region, divide it into two in the $x$-direction, then into two in the $y$-direction, then in the $x$ direction and so on.
The principle is that we are still dividing into two at each stage, even though we would be aiming to divide into 4 and use the information from four squares at each level to compute $Q$ for one level up.
Similarly, in 3D we would subdivide the region in the $x$ direction, then $y$ direction then $z$, and repeat.

This is equivalent to interleaving the bits of the $x$, $y$, and $z$ components of the electron position.
Interleaving the bits of the three coordinates and sorting yields the particles in Morton ordering.
When we do this we can perform exactly the same procedure as in the 1D case with the interleaved bits of $x,y,z$ replacing the bits of $x$.
At each stage we would add charge from two neighbouring regions rather than 8, but three levels give addition of the charges from 8 sub-cubes.
Algorithm \ref{alg:calccharg} is unchanged, except $\boxno{j,\ell}$ is now given by the bits up to $\ell$ of the electron register with bits for the $x$, $y$, and $z$ coordinates interleaved, and $\Levs$ would be the total number of bits (but three times the number of levels).

For the multipole information there is a translation used as well when aggregating the information for the boxes to the next level up.
In 2D or 3D this translation will be to the centre of the parent box, even though we are adding together information for each direction separately.
This means that the translation operation need only be performed once for each change in level, similar to the classical FMM.

Note that when we use the registers with $Q(j,\ell)$ for calculating the potential, we want $\ell$ to be the usual FMM level.
This means that we need to adjust the indexing, which just corresponds to a choice of labelling of the registers, rather than requiring quantum operations.
We assume that this has been done when using $Q(j,\ell)$ in later steps.

\subsection{Calculating the potential}
\label{sec:calcpotl}

The principle is to run through the registers for all electrons and copy the information for one region over to the neighbouring region.
This is similar to what we have already done for the charge, but when calculating the potential we have the further complication that it is more difficult to find the region.
When considering the charge, we are able to combine two regions in the $x$ direction, then $y$ direction, then $z$ direction.
For the potential we need to find particles in regions in the interaction list.

The geometry of the interaction list is illustrated in Fig.~\ref{fig:Interaction} for two dimensions.
The box of interest (in green) at level-$\ell$ is part of a level-$\ell-1$ box indicated with thick black lines.
The other level-$\ell$ boxes inside that box are neighbours, so we do not add the potential.
That is in contrast to the method for charge, where we add charges for sub-boxes.
The boxes where we do add the potentials, indicated in orange, are part of the neighbouring level-$\ell-1$ boxes.

Not only that, but we have a row of three of these larger boxes.
This means that they are not just contained in one level-$\ell-2$ box, but overflow into neighbouring boxes at that level.
Nevertheless, we can still use a similar approach for copying information from neighbouring regions, which we will initially explain for the 1D case.

\subsubsection{Accessing interaction list for 1D case}

For each electron, we have extra registers for one, two, and three positions over; we will call these $Q_{1}(j,\ell)$, $Q_{2}(j,\ell)$, and $Q_{3}(j,\ell)$.
We run through the electrons in sequence, and for each we check if the following electron corresponds to a new region; that is $\boxno{j,\ell}\ne \boxno{j+1,\ell}$.
The cases are then as follows.
\begin{enumerate}
    \item If $\boxno{j,\ell} = \boxno{j+1,\ell}$, then both electrons are in the same region, and we just directly copy over the charge information.
    That is, $Q_{k}(j+1,\ell) = Q_{k}(j,\ell)$ for $k\in\{1,2,3\}$.
    \item If $\boxno{j+1,\ell}= \boxno{j,\ell}+1$, then the next electron is in the next region, and we copy the charge information to the next electron's register for the neighbouring region; that is, $Q_{1}(j+1,\ell) = Q(j,\ell)$.
    We also copy the charge information for one location over to that for two locations over, and that for two regions over to that for three regions over,
    so $Q_{2}(j+1,\ell) = Q_{1}(j,\ell)$ and $Q_{3}(j+1,\ell) = Q_{2}(j,\ell)$.
    \item If $\boxno{j+1,n}= \boxno{j,n}+2$, then the next electron is two regions displaced, and we directly copy into the register for two regions over, so $Q_{2}(j+1,\ell) = Q(j,\ell)$.
    We also copy the information for one region over to that for three regions over, so $Q_{3}(j+1,\ell) = Q_1(j,\ell)$.
    We leave $Q_{1}(j+1,\ell)$ zeroed, because that region is missing.
    \item If $\boxno{j+1,n}= \boxno{j,n}+3$, then the next electron is three regions displaced, and we directly copy into the register for three regions over, so $Q_{3}(j+1,\ell) = Q(j,\ell)$.
    We leave both $Q_{1}(j+1,\ell)$ and $Q_{2}(j+1,\ell)$ zeroed, because those regions are missing.
    \item If $\boxno{j+1,\ell}>\boxno{j,\ell}+3$, then the next electron is more than three regions displaced, and $Q_{1}(j+1,\ell)$, $Q_{2}(j+1,\ell)$, and $Q_{3}(j+1,\ell)$ can be left zeroed.
\end{enumerate}

After doing this, we have associated with each electron, registers giving the charge information for the boxes to the left of it, which can be used in computing the potential.
We can then invert the above procedure to erase those working ancillas.
To describe the procedure in a more general case where we copy information from up to $K$ boxes over, we use registers $Q_{k}(j,\ell)$ for $k=1$ to $K$, and it is also convenient to use $Q_{0}(j,\ell)$ to denote the charge information for the box itself.
Then the general procedure for copying information from neighbours is as shown in Algorithm \ref{alg:copy}.

\begin{algorithm}
    \caption{Copying information for $K$-neighbours}\label{alg:copy}
    \For{$j\gets 0$ \KwTo $\eta-2$}{
        \uIf{$\boxno{j+1,\ell} = \boxno{j,\ell}$}{
            \For{$k\gets 1$ \KwTo $K$}{
                $Q_{k}(j+1,\ell) \gets Q_{k}(j,\ell)$}
        }
        \ElseIf{$\boxno{j+1,\ell} \le \boxno{j,\ell}+K$}{
            $\Delta k \gets \boxno{j+1,\ell} - \boxno{j,\ell}$\\
            \For{$k\gets \Delta k$ \KwTo $K$}{
                $Q_{k}(j+1,\ell) \gets Q_{k-\Delta k}(j,\ell)$}
        }
    }
\end{algorithm}

In this algorithm, we run through the registers from $0$ to $\eta-2$, with the requirement that the registers are sorted so $\boxno{j,\ell}$ (bits up to $\ell$ of the position) are in ascending order.
Then $\boxno{j+1,\ell} = \boxno{j,\ell}$ implies that the next register is in the same box, so we copy over the information for the neighbours of that box.
If $\boxno{j+1,\ell} \le \boxno{j,\ell}+K$, then we have stepped over to another box, but by no more than $K$.
Then we copy over the appropriate neighbour information, skipping some if $\Delta k>1$, which means that there were empty boxes skipped over in going from particle $j$ to $j+1$.

To access information from all boxes in the interaction list, we can also perform the copying in the reverse direction.
Although copying in the reverse direction is needed to access information from all boxes in the interaction list, it is not necessary to calculate the potential.
This is because, for each box $a$ in the interaction list of box $b$, $b$ will be in the interaction list of $a$.
Copying the information in the reverse direction will only give contributions to the total potential energy that have already been accounted for, so is not needed.

Algorithm \ref{alg:copy} can also be used to calculate the charge / multipole information in 3D, instead of using the intermediate levels in the simplified method described above.
That is, replace steps 2(a) and 2(b) in Algorithm \ref{alg:calccharg} with Algorithm \ref{alg:copy} (with $K=7$) to copy information from registers that are in the same box at level $\ell$ but different at level $\ell+1$.
This will retrieve information for the 8 boxes at level $\ell+1$ necessary to calculate the charge / multipole information for each box at level $\ell$.

\subsubsection{Accessing interaction list for higher dimensions}

In order to obtain the information from the boxes in the interaction list in two or three dimensions, we can use a similar approach interleaving bits as we use for charge, but it is more challenging because the boxes in the interaction list need to be accessed by taking a 1D path through this higher-dimensional space.
To address this issue, we use Morton ordering, as well as bit shifts.
In particular, the Morton ordering is obtained by interleaving the bits for the $x$, $y$, and $z$ directions and sorting.

In Figure \ref{fig:Morton} the level-$\ell-2$ box is shown in red, with four different alternatives for where it can be situated relative to the box of interest, $b$.
The Morton ordering means that iterating through the particles in the sorted list will run through the boxes in the sequence indicated by the blue arrows.
Therefore, we can use the approach above for 1D, but use temporary registers to account for boxes in the interaction list displaced by up to 15 positions (in 2D) or 63 positions (in 3D).

Note that the boxes in the interaction list may be before or after the target box $b$ in the ordering, so we would need to go through the sorted list both forwards and backwards to access information from all boxes in the interaction list.
However, there is a similar consideration here as for the 1D case.
Because $a\in\inter(b) \implies b\in \inter(a)$, going through the sorted list backwards would only yield contributions to the potential that have already been accounted for when running through the list forwards, and so the reverse direction is not needed.

\begin{figure}
    \centering
\sidesubfloat[]{\includegraphics[width = 0.3\textwidth]{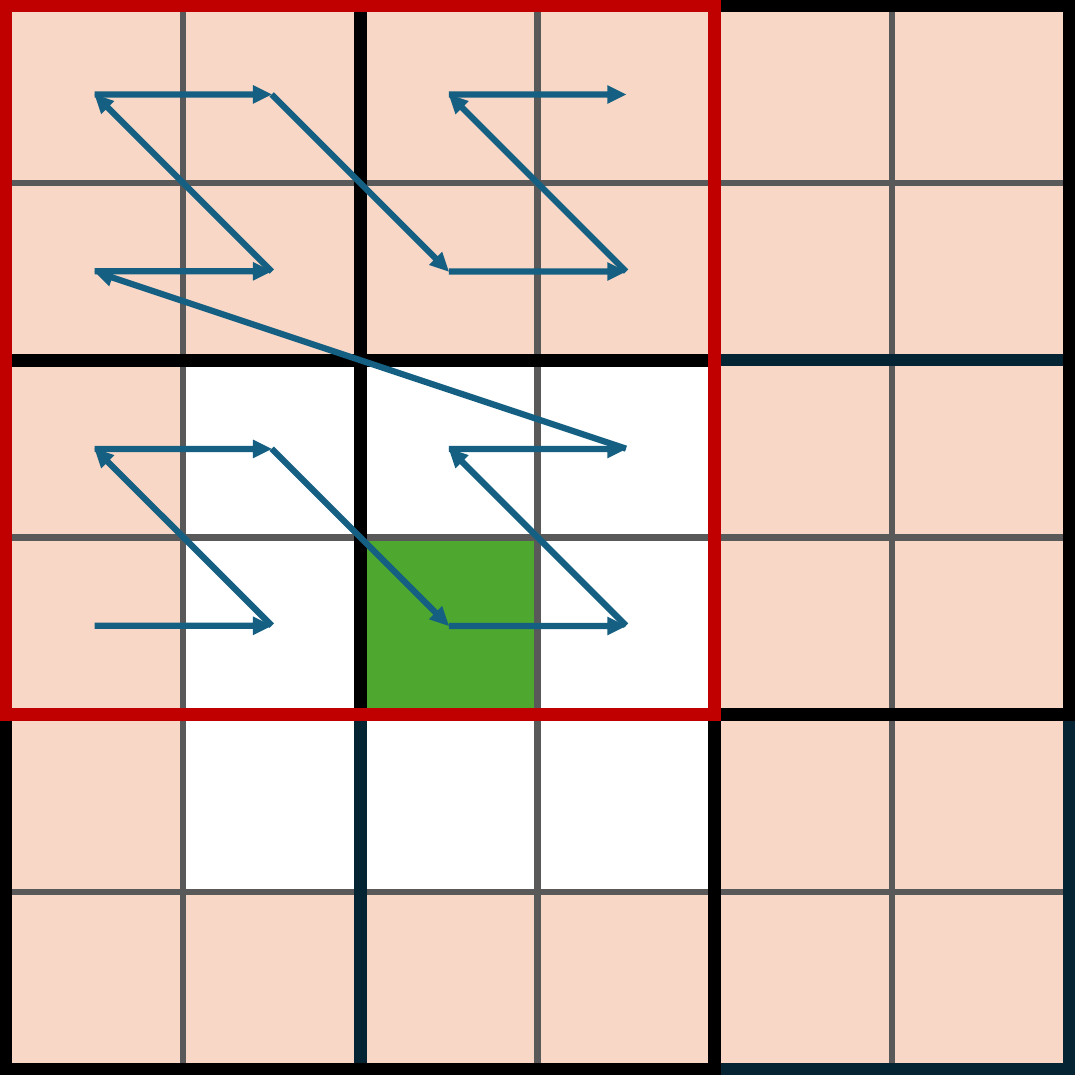}\label{fig:a}}\hspace{2.5mm}
   \sidesubfloat[]{\includegraphics[width = 0.3\textwidth]{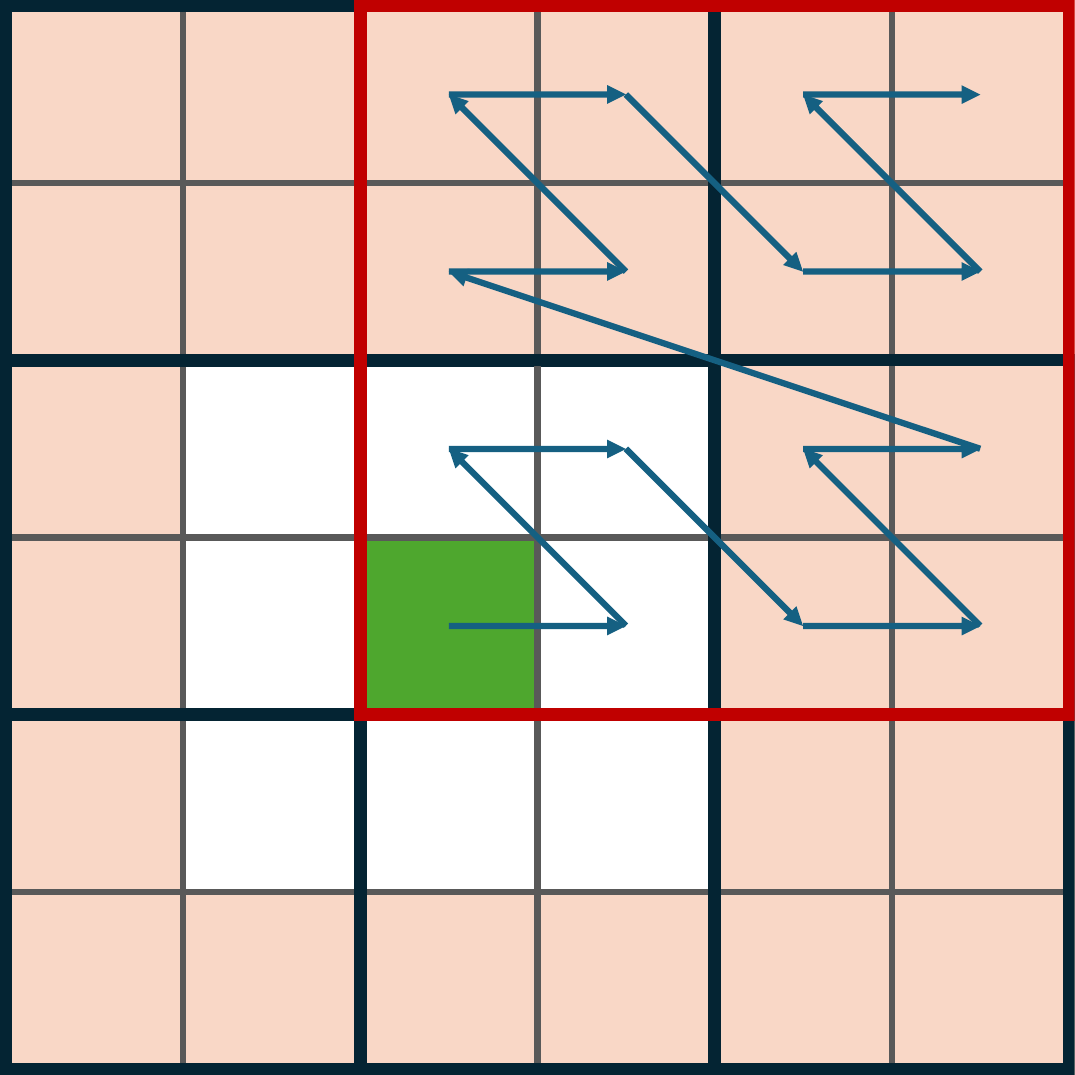}\label{fig:b}}\\
   \vspace{2mm}
    \sidesubfloat[]{\includegraphics[width = 0.3\textwidth]{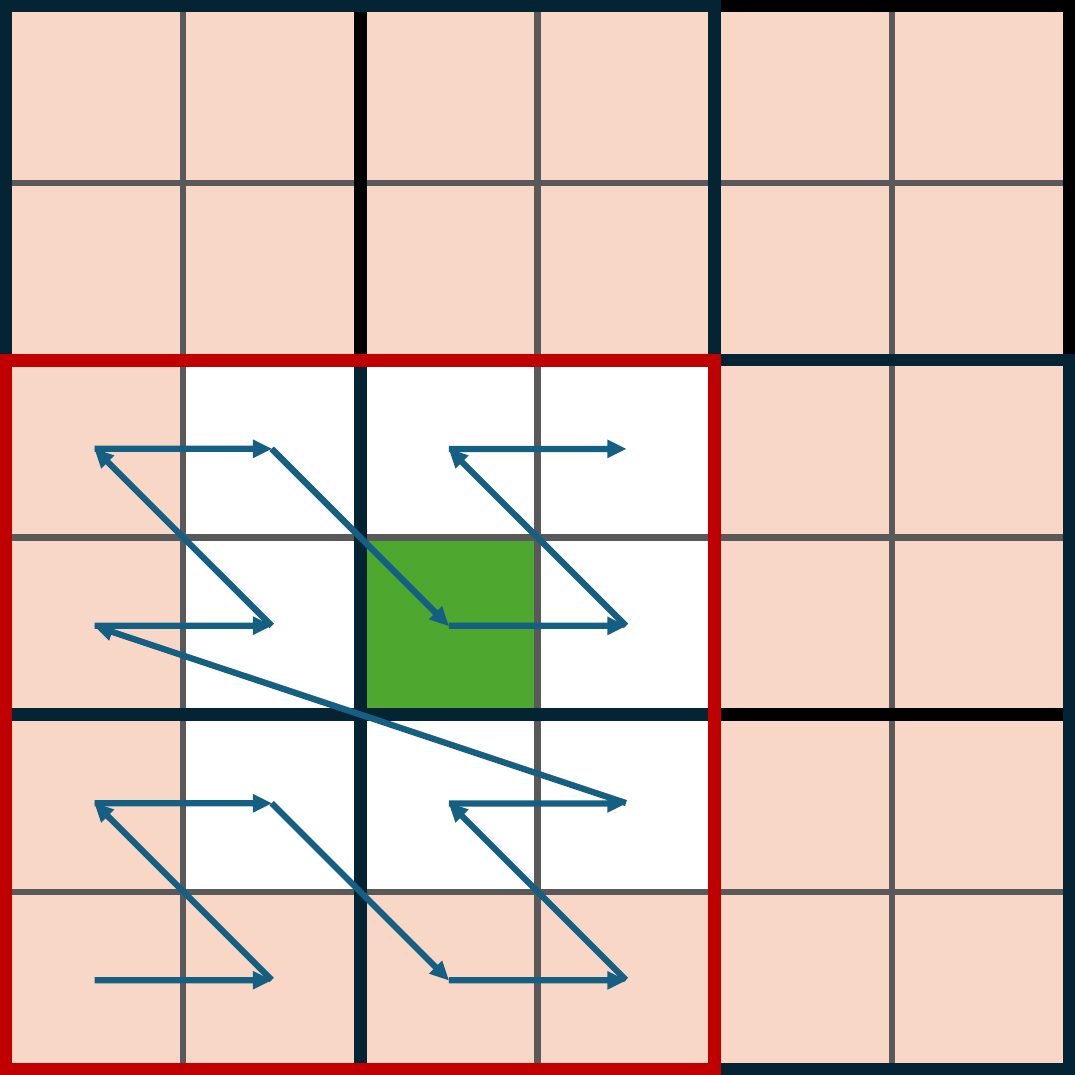}\label{fig:c}}\hspace{2.5mm}
    \sidesubfloat[]{\includegraphics[width = 0.3\textwidth]{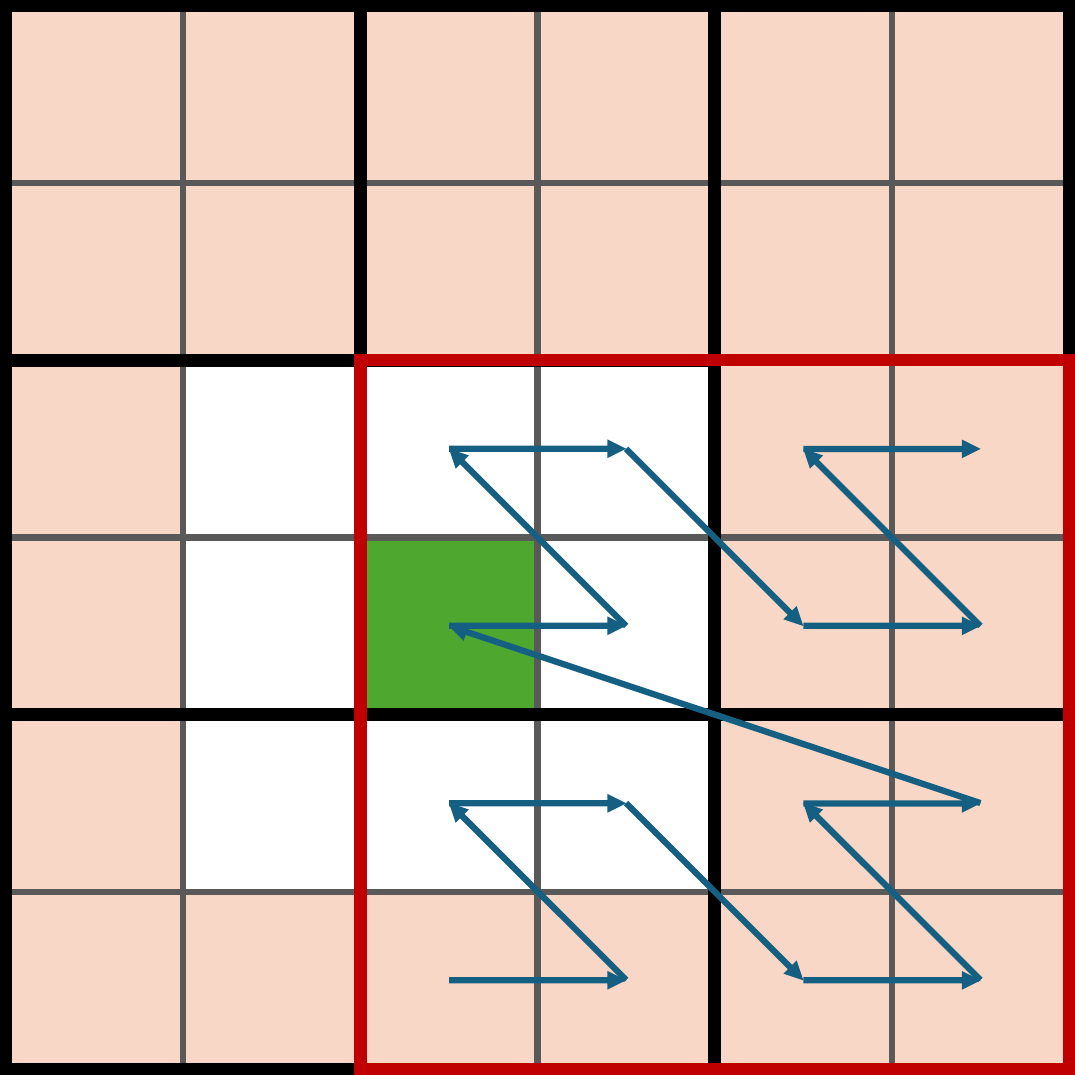}\label{fig:d}}
    \caption{The interaction list and neighbours used for FMM.
    The green square is the box $b$ of interest, the white squares are neighbours of $b$, and the boxes in the interaction list are shown in orange.
    The heavy black lines indicate the boxes for the next level up.
    The red lines indicate the box two levels up containing the box of interest, with parts (a) to (d) indicating the four alternatives.
    The blue zig-zag lines indicate the Morton ordering of boxes within that higher-level box.}
    \label{fig:Morton}
\end{figure}

Using just a single Morton ordering does not guarantee that all the boxes in the interaction list are within a fixed distance from $b$ in that ordering.
This is an issue that has been considered in the classical computing literature, and can be solved by using multiple orderings.
For example, Refs.~\cite{Liao2001,Barkalov2022} use multiple Hilbert orderings, and Ref.~\cite{LikANN} uses multiple Morton orderings.
Here, we use a slightly different approach to those, which is more easily implemented in a quantum algorithm.

In particular, the Morton ordering is more easily implemented than the Hilbert ordering, because it can be achieved by simply interleaving the bits of $x$, $y$, and $z$, then applying the quantum sort.
To achieve the multiple orderings, we apply combinations of increments in the $x$, $y$, and $z$ directions, and sort again.
For a level-$\ell-2$ box, bit number $\ell-3$ is identical for all positions within the box.
That is, bit number $\ell-3$ is identical for all $x$-coordinates, bit number $\ell-3$ is identical for all $y$-coordinates, and so on.
For example, if $\ell-2=2$ (the first subdivision of the complete region), then bit number $\ell-3=1$ is the first (most significant) bit.

Now the effective level-$\ell-2$ box can be changed by incrementing the coordinates.
If the level-$\ell-2$ box is that shown in Figure \ref{fig:Morton}(a), then adding 1 to bit number $\ell-2$ for the $x$-coordinate will shift the effective level-$\ell-2$ box over to that shown in Figure \ref{fig:Morton}(b).
For example, if $\ell-2=2$ and the first 3 bits of the $x$-coordinates are
\begin{equation}
    000,~ 001,~ 010,~ 011,~ 100,~ 101,
\end{equation}
then you can see that the first bit is the same for the first four (corresponding to the first four level-$\ell$ boxes in the diagram).
Adding 1 to bit number $\ell-2=2$ gives
\begin{equation}
    010,~ 011,~ 100,~ 101,~ 110,~ 111,
\end{equation}
so the four boxes on the right have the same first bit, and are within the same effective level-$\ell-2$ box.

In this way, by the 4 combinations of incrementing $x$ and $y$ (or $x$, $y$, and $z$ in the 3D case), the effective level-$\ell-2$ box can be shifted to each of the possible locations.
This will be true regardless of where the initial (unshifted) position of the level-$\ell-2$ box is.
For example, if the position is that illustrated in Figure \ref{fig:Morton}(b), then the increment in the $x$-coordinate will shift the effective box to that shown in Figure \ref{fig:Morton}(a).

For each shift, the bits of $x$, $y$, and $z$ are interleaved and the sort is performed to give the Morton ordering with a shift.
Performing this with each combination of shifts guarantees that each box in the interaction list is within a fixed distance of $4^d-1$ for dimension $d$ in at least one ordering.
In particular, the result can be given as follows, with the notation $[a,b]$ indicating the set corresponding to the range of integers.

\begin{lemma}[Morton orderings with shifts]
\label{lem:Mort}
    For $\vb{p}\in [0,2^n-1]^d$, let
    \begin{equation}
        D_p := \left\{ \vb q\in [0,2^n-1]^d ~\middle|~  |\lfloor p_j/2 \rfloor -\lfloor q_j/2 \rfloor |\le 1,~ \forall j\in[1,d]  \right\}\, ,
    \end{equation}
    and $M: \mathbb{N}^d \to \mathbb{N}$ be the function applying the Morton ordering to a vector of integers.
    Then for all $\vb q\in D_p$, there exists a vector of increments $\vb z \in \{0,2\}^d$ such that 
    \begin{equation}\label{eq:Mortdif}
        |M(\vb p+\vb z \bmod 2^n)-M(\vb q+\vb z \bmod 2^n)| \le 4^d-1 \, .
    \end{equation}
\end{lemma}

\begin{proof}
    First, note that points satisfying $\lfloor \vb p/2^m\rfloor=\lfloor \vb q/2^m\rfloor$ will satisfy $\lfloor M(\vb p)/2^{dm}\rfloor=\lfloor M(\vb q)/2^{dm}\rfloor$.
    This is because the set of points satisfying $\lfloor \vb p/2^m\rfloor=\lfloor \vb q/2^m\rfloor$ are those within a cube of side length $2^m$ corresponding to the first $n-m$ bits of each coordinate of $\vb p$ being equal to those of $\vb q$.
    The Morton ordering groups all $2^{dm}$ points with these same initial bits together, because it corresponds to sorting with the bits interleaved.
    Because there are $2^{dm}$ points, they can be separated by no more than $2^{dm}-1$ in the Morton ordering.

    Now, given a point $\vb q\in D_p$, take
    \begin{equation}
        \vb z = 2|\lfloor \vb p/4 \rfloor -\lfloor \vb q/4 \rfloor | \, .
    \end{equation}
    Since $|\lfloor p_j/2 \rfloor -\lfloor q_j/2 \rfloor |\le 1$, 
    it must also be true that $|\lfloor p_j/4 \rfloor -\lfloor q_j/4 \rfloor |\le 1$, and so
    this vector is within the required set $\{0,2\}^d$.
    For any given $j$, if $z_j=0$, then $\lfloor p_j/4 \rfloor =\lfloor q_j/4 \rfloor$.
    Then it is clear that
    \begin{equation}\label{eq:peqq}
        \lfloor (p_j+z_j \bmod 2^n)/4\rfloor = \lfloor p_j/4\rfloor = \lfloor q_j/4\rfloor = \lfloor (q_j+z_j \bmod 2^n)/4\rfloor \, .
    \end{equation}
    On the other hand, if $z_j=2$, then put $p_j=4a+r$ and $q_j=4b+s$ with $r,s\in \{0,1,2,3\}$.
    For $p_j<q_j$, then $b=a+1$ and $r\in\{2,3\}$, $s\in \{0,1\}$.
    This restriction on $r,s$ is required so that $|\lfloor p_j/2 \rfloor -\lfloor q_j/2 \rfloor |\le 1$ for the requirement that $\vb q\in D_p$.
    Then
    \begin{align}
        p_j+z_j &= 4a+r + 2 = 4(a+1) + r-2 = 4b + r-2 \, , \\
        q_j+z_j &= 4b+s + 2 \, .
    \end{align}
    The conditions $r\in\{2,3\}$, $s\in \{0,1\}$ imply that $r-2\in\{0,1\}$, $s+2\in \{2,3\}$.
    Note that neither $p_j+z_j$ or $q_j+z_j$ can increase above $2^n-1$, and so the modular addition does not affect the outcome, and
    \begin{align}
        \lfloor (p_j+z_j \bmod 2^n)/4\rfloor &= \lfloor (p_j+z_j)/4\rfloor \nn
        &= \lfloor (4b + r-2)/4\rfloor \nn
        &= b \nn
        &= \lfloor (4b + s+2)/4\rfloor \nn
        &= \lfloor (q_j+z_j)/4\rfloor \nn
        &= \lfloor (q_j+z_j \bmod 2^n)/4\rfloor \, .
    \end{align}
    The reasoning for $z_j=2$ and $p_j>q_j$ is identical with the roles of $p_j$ and $q_j$ reversed.
    Thus in all cases
    \begin{equation}
        \lfloor (p_j+z_j \bmod 2^n)/4\rfloor = \lfloor (q_j+z_j \bmod 2^n)/4\rfloor \, .
    \end{equation}

    In turn, using the result given at the start of the proof, this implies that $\lfloor M(\vb p+\vb z)/4^{d}\rfloor=\lfloor M(\vb q+\vb z)/4^{d}\rfloor$, and the points $M(\vb p+\vb z \bmod 2^n)$ and $M(\vb q+\vb z \bmod 2^n)$ cannot be separated by any more than $4^d-1$, which is the result to be proven for the Lemma.
\end{proof}

The interpretation of this Lemma is that $\vb p$ is the vector for a box number at level $\ell$ for dimension $d$, and $D_p$ is the set of boxes that are neighbours of $\vb p$ or in the interaction list (as well as $\vb p$ itself).
This means that it is within a level-$\ell-1$ box (the reason for the division by 2 and floor) displaced by no more than one position in any direction.
Less significant bits for positions within that box number are not included in this Lemma.
Then $\vb z$ is the vector of displacements by 2, corresponding to shifting by the size of a level-$\ell-1$ box.
The additions are modulo the size of the entire region of $2^n$, then Eq.~\eqref{eq:Mortdif} bounds the difference between $\vb p$ and $\vb q$ in the shifted Morton ordering.

This choice of shifted Morton orderings then enables us to copy to each box the data from all boxes in its interaction list for the calculation of the potential.
The complete procedure is given in Algorithm \ref{alg:potl}, which is the quantum implementation of loop 3 of the FMM algorithm in Algorithm \ref{alg:FMMagg}.
At each level, we apply the increment and the quantum sort according to the Morton ordering.
With that sort, we go through the list of particles, copying the information from neighbouring boxes at level $\ell$.
As discussed above, we need only go through the list in one direction.

\begin{algorithm}
\caption{Calculation of potential in dimension $d$}\label{alg:potl}
$K\gets 4^d-1$\\
$V\gets 0$\\
\For{$\ell \gets 3$ \KwTo $\Levs$}{
    \For{$\vb z\in \{0,2\}^d$}{
        Add $\vb z$ to box number at level $\ell$ and sort according to Morton ordering.\\
        Apply Algorithm \ref{alg:copy} to copy charge information. \\
        \For{$j\gets 0$ \KwTo $\eta-1$}{
            $V_b \leftarrow V_{\parent(b)}$ \\
            \For{$k\gets 1$ \KwTo $K$}{
                \If{$Q_{k}(j,\ell)$ contains charge information for a box in the interaction list not already added}{
                    $V_b \leftarrow V_b + \Phi(\vb{c}_a,\vb{c}_b)Q_{k}(j,\ell)$
                }
            }
        }
        Invert Algorithm \ref{alg:copy} to erase $Q_{k}(j,\ell)$ registers.
        Invert Morton ordering sort and subtract $\vb z$.
    }
}
\end{algorithm}

At each step, we check whether the potential was added at a previous step, because the set of Morton orderings results in the charge information for each box in the interaction list being included multiple times.
In Algorithm \ref{alg:potl}, the quantity $\vb{c}_b$ is the position of the centre of the level-$\ell$ box that particle $j$ is in, and $\vb{c}_a$ is the centre of the box that the information in $Q_{k}(j,\ell)$ corresponds to.
The box potential $V_b$ is stored in a register associated with each particle as well.
Because registers associated with each particle contain the information for boxes it is contained in at each level, the information from the parent box may be easily copied.
That is, we do not need to locate information in registers that are associated with other particles for $V_b \leftarrow V_{\parent(b)}$.

The remaining part needed is the direct calculation of the sum over pairwise interactions with particles in neighbouring boxes for $\ell=\Levs$.
Note that we do not sum over multiple particles in the same box, because we take level $\Levs$ to correspond to the maximum resolution of the particle positions, so only one electron may occupy the box.
The additional sum of the contributions to the potential from particles in the neighbouring boxes at level $\Levs$ is easily accounted for by including neighbouring boxes in the interaction list in Algorithm \ref{alg:potl}.
The final sum over $q_iV_i$ to obtain the potential energy $V$ is easily obtained by summing over the values of $V_b$ that have been calculated at level $\ell=\Levs$.
In this way, this algorithm effectively gives the quantum implementation of the downward pass of the FMM in Algorithm \ref{alg:FMMagg}, with the quantum implementation of the upward pass of Algorithm \ref{alg:FMM} being given in Algorithm \ref{alg:calccharg}.

Accounting for the electron-nuclear potential is more complicated for the adaptive case than for the non-adaptive case.
It is no longer possible to choose the nuclei to sum over classically, because the box number is governed by the value in a quantum register, and so the correct nuclei to consider are chosen by that register.
A simple way to account for the nuclei is to include registers for the nuclear positions in the list with the electron positions, and similarly associate registers for the multipole information at each level for each nucleus.
Because the number of nuclei is no more than the number of electrons in cases of interest, the inclusion of nuclei does not change the asymptotic complexity.

\section{Quantifying the complexity}
\label{sec:complex}
Now we give the overall complexity of the simulation with the quantum FMM.
We will analyse only the adaptive protocol, as it does not require any assumption on the distribution of particles.
We first consider the complexity when only using the charge to calculate the potential, then consider the adjustment to the complexity for the multipole information.
We also just analyse the complexity for calculating the electron-electron potential, because the nuclei do not change the complexity.
In Algorithm \ref{alg:calccharg} for calculating the charge information, there are $\Levs-1$ steps, and for each there are $\eta-1$ steps of copying charge information, as well as $\eta$ steps of adding charges together.
The complexity of the copying or addition is proportional to the number of bits for the charge which is at most $\lceil \log\eta\rceil$.
This gives a complexity of 
\begin{equation}\label{eq:step1comp}
    \cO(\eta \log \eta \log N) \, ,
\end{equation}
where the factor of $\eta$ is from the iteration over particles, the factor of $\log\eta$ from the arithmetic, and $\log N$ from the number of levels.
In addition, the sorting to Morton ordering gives the same contribution to the complexity.
This is because sorting $\eta$ items with optimal sorting networks requires $\cO(\eta \log \eta)$ steps.
Each step has complexity $\cO(\log N)$ because that is the size of the items to be sorted (the particle positions), giving the same complexity as in Eq.~\eqref{eq:step1comp}.

Next we consider the complexity of Algorithm \ref{alg:potl} for computing the potential.
The contributions to the complexity are as follows.
\begin{enumerate}
    \item There are $\cO(\log N)$ sorts to be performed, with each being a sort of $\eta$ items and therefore having complexity $\cO(\eta \log \eta)$ with optimal sorting networks.
    The items to be sorted include only the particle positions and the corresponding box charge for level $\ell$, and so is of size $\cO(\log N)$.
    Note that we only need the charge information for level $\ell$, so do not need to sort the charge information for the other levels at this step.
    As a result the complexity of the sorts is
\begin{equation}\label{eq:potlsorts}
    \cO(\eta \log \eta \log^2 N) \, .
\end{equation}
\item We need to use Algorithm \ref{alg:copy} to copy the charge information.
This is done $\cO(\log N)$ times, and each time there are $\eta$ steps with $\cO(\log\eta)$ information to be copied.
As a result the complexity is the same as in Eq.~\eqref{eq:step1comp}.
\item There are $\cO(\eta\log N)$ steps where we need to calculate the addition to the potential.
There are $\cO(1)$ boxes in the interaction list, and so $\Phi(\vb{c}_a,\vb{c}_b)$ can be determined with a QROM with complexity $\cO(1)$ Toffolis, or $\cO(\log(1/\epsilon))$ total gates accounting for the precision that the potential is given to.
The multiplications by $Q_{k}(j,\ell)$ have complexity $\cO(\log \eta \log(1/\epsilon))$, and the addition has negligible complexity in comparison.
There is also a check whether the information is from a box in the interaction list, which is based on the coordinate and can be performed with $\cO(\log N)$ complexity.
As a result, the total complexity of this part is
\begin{equation}\label{eq:mults}
    \cO(\eta \log N [\log \eta \log(1/\epsilon)+\log N]) \, .
\end{equation}
In practice, the cost of multiplication is significantly larger than the cost of checking whether the box is in the interaction list, so the $\log N$ in the square brackets may be omitted.

\item The final sum over $q_iV_i$ is just a sum over the local potentials and it has complexity
\begin{equation}
    \cO(\eta\log(1/\epsilon)) \, .
\end{equation}
\end{enumerate}
The main contributions to the complexity are from Eq.~\eqref{eq:mults} and Eq.~\eqref{eq:potlsorts}.
In this costing, the number of bits for the arithmetic is taken to be $\cO(\log(1/\epsilon))$, so there is no need to take the number of bits for each component of the coordinates to be more than $\cO(\log(1/\epsilon))$.
This corresponds to a choice of $N$ such that $\log N=\cO(\log(1/\epsilon))$\footnote{Note that it is usually found that the error due to the finite size of basis sets scales as $\cO(1/N)$
\cite{gruneis2013explicitly,hattig2011explicit,Harl2008,shepherd2012convergence,helgaker1997basis_b,klopper1995ab2,Halkier1998b}, so it is sufficient to take $N=\cO(1/\epsilon)$.
That also implies $\log N=\cO(\log(1/\epsilon))$.}.
In that case the complexity of the whole procedure can be given as $\cO(\eta \log \eta \log N \log(1/\epsilon))$.

The leading contributions to the qubits used are as follows.
\begin{enumerate}
    \item The number of qubits to store the system state is
\begin{equation}
    \cO(\eta \log N) \, .
\end{equation}
\item The qubits used for the charge information at all levels is
\begin{equation}
    \cO(\eta \log \eta \log N) \, .
\end{equation}
This is because there are $\cO(\log N)$ levels, and each uses $\cO(\log \eta)$ qubits for the charge.
\item The potentials $V_b$ are stored for each level with $\cO(\log(1/\epsilon))$ qubits each, for a total
\begin{equation}
    \cO(\eta \log N \log(1/\epsilon)) \, .
\end{equation}
\item There are temporary qubits used for the coherent sorts.
Because these record the results of inequality tests, and there are $\cO(\eta \log \eta)$ of these tests, the total number of qubits is
\begin{equation}
    \cO(\eta \log \eta) \, .
\end{equation}
\end{enumerate}
The leading contribution to the logical qubits used is $\cO(\eta \log N \log(1/\epsilon))$ to store the potentials $V_b$.

More generally, for the multipole expansion, the order needed for precision $\epsilon$ is $\trunc=\cO(\log(1/\epsilon))$.
The number of multipole moments scales as the square of the order, and each needs $\cO(\log(1/\epsilon))$ bits.
As a result, rather than just the charge being given with $\lceil\log \eta\rceil$ bits, the multipole information needs a number of bits $\cO(\trunc^2\log(1/\epsilon))=\cO(\log^3(1/\epsilon))$.
Similarly, $V_b$ is replaced with a local expansion needing a number of bits $\cO(\log^3(1/\epsilon))$ (see Appendix~\ref{app:additional_fmm}).
The total
number of qubits used can therefore be given as
\begin{equation}
    \cO(\eta \log N  \log^3(1/\epsilon)) \, .
\end{equation}

Similarly, the operations on the multipole information involve $\cO(\trunc^2)$ arithmetic operations, which are primarily multiplications, so have complexity $\cO(\log^2(1/\epsilon))$.
These operations are performed for each electron at $\log N$ levels, giving the complexity of the quantum FMM as
\begin{equation}
    \cO(\eta \log N \log^4(1/\epsilon)) \, .
\end{equation}
This complexity is true both for the calculation of the multipole information by combining that for individual boxes, and the calculation of the potential (see Appendix \ref{app:M2P}).

For the overall complexity, we need to account for the number of Trotter steps as well.
This is accounted for in Ref.~\cite{LowPRX2022}, which we summarise here.
A real-space grid Hamiltonian may be given as
\begin{align}\label{eq:ueg}
H = \sum_{j, k,\sigma}\tau_{j,k}a_{j,\sigma}^{\dagger}a_{k,\sigma} + \sum_{l, m,\sigma,\tau}\nu_{l,m}a_{l,\sigma}^{\dagger} a_{l,\sigma}a_{m, \tau}^{\dagger} a_{m, \tau} \, ,
\end{align}
where $a_{j,\sigma}^\dagger$ and $a_{j,\sigma}$ are creation and annihilation operators, $\{j, k, l, m\}$ are orbital indices, and $\{\sigma, \tau\}$ are spin indices.
The spectral norm error in a fixed number particle manifold for an order-$k$ product formula $S_k(t)$ can be estimated as \cite{LowPRX2022}
\begin{equation}\label{eq:Low_bound}
    \norm{S_k(t)-e^{-itH}}_{W_\eta}=\cO\left((\norm{\nu}_{1,[\eta]}+\norm{\tau}_{1})^{k-1}\norm{\tau}_1\norm{\nu}_{1,[\eta]}\eta\, t^{k+1}\right) \, ,
\end{equation}
where the norms are defined as 
\begin{align}
\|\nu\|_{1,\left[\eta\right]} &= \max_{j} \max_{k_{1} < ... < k_{\eta}} \left(|v_{j,k_{1}}| + ... + |v_{j, k_{\eta}}| \right) , \\
\|\tau\|_{1} &= \max_{j} \sum_{k}|\tau_{j,k}| \, .
\end{align}
The number of time steps needed for a product formula of order $k$ is then
\begin{equation}
    \cO\left( t^{1+1/k} (\norm{\nu}_{1,[\eta]}+\norm{\tau}_{1})^{1-1/k}(\norm{\tau}_1\norm{\nu}_{1,[\eta]}\eta/\epsilon)^{1/k}\right) \, .
\end{equation}
The scaling of the norms is
\begin{align}
    \norm{\nu}_{1,[\eta]} &= \cO \left( \frac{\eta^{2/3} N^{1/3}}{\Omega^{1/3}} \right)\, , \\
    \norm{\tau}_{1} &= \cO \left( \frac{N^{2/3}}{\Omega^{2/3}} \right) \, .
\end{align}
This then gives the number of steps as
\begin{equation}\label{eq:complex}
    \cO\left( t^{1+1/k} \left(\frac{\eta^{2/3} N^{1/3}}{\Omega^{1/3}}+\frac{N^{2/3}}{\Omega^{2/3}}\right)^{1-1/k}\left(\frac{\eta^{5/3} N}{\Omega\epsilon}\right)^{1/k}\right) \, .
\end{equation}
Since $k$ may be chosen arbitrarily large (though in practice a moderate value will be optimal), we can give the complexity with $1/k$ replaced by $o(1)$, which gives
\begin{equation}
    t \left(\frac{\eta^{2/3} N^{1/3}}{\Omega^{1/3}}+\frac{N^{2/3}}{\Omega^{2/3}}\right)\left(\frac{\eta N t}{\Omega\epsilon}\right)^{o(1)} \, .
\end{equation}
The $\cO$ is not given for the entire expression, because the fact that this is an asymptotic scaling is accounted for in the $o(1)$.
The factors of $\log N$ and $\mathrm{polylog}(1/\epsilon)$ for the complexity of the FMM calculation can be included with the $N^{o(1)}$ and $1/\epsilon^{o(1)}$.
Therefore, multiplying by the Toffoli complexity of the FMM calculation gives the Toffoli complexity of the complete simulation as
\begin{equation}
    t \left(\frac{\eta^{5/3} N^{1/3}}{\Omega^{1/3}}+\frac{\eta N^{2/3}}{\Omega^{2/3}}\right)\left(\frac{\eta N t}{\Omega\epsilon}\right)^{o(1)} \, ,
\end{equation}
as given above in Eq.~\eqref{eq:complexity}.
A common assumption to simplify the complexity is the thermodynamic limit, where $\eta\propto\Omega$.
This gives the complexity as
\begin{equation}
    t \left(\eta^{4/3} N^{1/3}+\eta^{1/3} N^{2/3}\right)\left(\frac{\eta N t}{\epsilon}\right)^{o(1)} \, ,
\end{equation}
which is the complexity stated in the abstract.
In this expression $\eta^{o(1)}$ is included because there is a factor of $(\eta^{5/3}/\Omega)^{1/k}$ in Eq.~\eqref{eq:complex}, which gives $\eta^{2/3k}$ for $\eta\propto\Omega$.

\section{Conclusion}
\label{sec:conc}
The need to perform a double sum with complexity $\cO(\eta^2)$ to compute the potential energy has been a long-standing bottleneck in simulation of quantum chemistry based on product formulae.
The fast multipole method provides a speedup to $\widetilde\cO(\eta)$ in classical computing, but the difficulty in translating the method to quantum algorithms is that it requires data accesses at positions according to the values in data registers, particularly for the adaptive case.
That would negate the speedup due to the overhead in quantum data access.

Here, we have provided a method to overcome this problem, yielding a quantum algorithm for FMM with complexity close to linear in the number of particles $\eta$.
Our method uses a similar principle as the quantum sort, where the data accesses are made in a fixed sequence of locations.
The procedure is considerably more difficult than the quantum sort, because it is also necessary to access data from boxes in the interaction list.

To overcome that problem, we combined the quantum sort with a set of shifted Morton orderings, in order to provide a moderate set of sequences which enable access to all boxes in the interaction list.
To access this information we associated box information with each particle, and used a procedure for copying the data from neighbouring boxes.

The result of this advance is that the overall complexity scaling is better than in prior work.
For grid-based approaches, large numbers of grid points (or plane waves) are required in order to provide accurate simulation, such that $N \gg \eta$.
In this limit, the best complexity is achieved by using interaction picture simulation, which provides a $N^{1/3}$ scaling \cite{BabbushContinuum}.
That method however has an additional factor of $\eta^3$, and so the method presented here has the best complexity for $N<\eta^7$, which is the typical range considered for realistic examples.
For more comparable methods, such as the first-quantised product formula approach in Ref.~\cite{RubinPNAS24}, we improve over the complexity by a factor of $\eta$ (up to logarithmic factors).

For these results on the FMM it is important that we are considering the first-quantised representation.
For second quantisation the double sum may be avoided in a much simpler way using the approach of Ref.~\cite{low2019hamiltoniansimulationinteractionpicture}, so there is no advantage to using the FMM.
We provide better performance than the second quantisation approach, because second quantisation yields a complexity at least linear in $N$ for the sum, which is much larger for the typical scenario where $N \gg \eta$.

In this work we have focused on showing that it is possible to achieve the asymptotic scaling associated with the FMM.
In practice, there are large constant factors that mean thousands of electrons are needed for FMM to provide a speedup.
There are further overheads involved in the quantum implementation, which could increase the regime for a speedup by an order of magnitude.
Moreover, the implementation uses a large amount of data, around 160 million logical qubits for the example of $\eta=4000$ electrons.
That is for the memory-efficient form in Appendix \ref{app:M2P}, where most of the storage is for multipole information.
In comparison, second-quantised representations would require larger numbers of qubits for large $N$, since a qubit is needed for every orbital.
Moreover, the method in Ref.~\cite{low2019hamiltoniansimulationinteractionpicture} has a further space overhead of a factor of $\log(1/\epsilon)$ for the method of calculating the potential.

Therefore, any physical implementation will require access to a quantum computer with a number of physical qubits exceeding $\sim 10^8$. The scale of this overhead depends on assumptions on the hardware platform as well as the error correction scheme used. For example, 
an implementation in the surface code would require thousands of physical qubits per logical qubit, but this can be improved upon by using high-rate LDPC codes, such as lifted-product (LP) codes, which are far more compact
\cite{Panteleev2021degeneratequantum}.
For example, LP memory codes can have rates as high as $0.285$ \cite{cain2026shorsalgorithmpossible10000}, implying $\sim 560$ million physical qubits for data storage, with additional qubits needed for fault-tolerant execution of logical gates.
That is large, though comparable to early estimates for simulation of FeMoco \cite{Elucidating}. 
Similarly, further improvements of the implementation here could significantly reduce the qubit requirements, positioning the FMM as a viable algorithm for later generations of fault-tolerant devices.
The coherent sorts used for data retrieval move data between a wide range of locations.
That means a physical implementation would need more than nearest-neighbour connectivity to avoid data-movement overheads, though that would already be needed for LP codes.

However, the focus in this work is on showing that it is possible to achieve the $\cO(\eta)$ complexity at all, and reducing the overheads is a topic for future work.
In particular, we have used an early form of the FMM for clarity of exposition \cite{greengard1988efficient,greengard1997new}, and there have been many advances since then.
For example, the Tucker decomposition \cite{Tucker} and reduced-rank approximations \cite{Hesford2011} can reduce the data needed for the FMM.
Recent work has shown a ``skeleton'' representation of the far field can simplify the implementation \cite{Skeleton}.
FMM code developed for GPUs, such as PKDGRAV3, is designed to minimise data movement as well as reduce memory usage \cite{potter2017pkdgrav3}.

There are also potential improvements that can be made to the quantum implementation to improve the efficiency.
Note that redundant information for all boxes is included by associating box information for each electron register.
It is possible that further improvements could be made by using only a single copy of information for each box (a principle used in Ref.~\cite{Stetina2025}).
There is also redundancy in our usage of multiple Morton orderings, which results in the information from boxes in the interaction list being accessed multiple times.
Potentially, alternatives such as Hilbert orderings \cite{Liao2001,Barkalov2022} or different shifts in the Morton ordering \cite{LikANN} could provide better performance.

The overheads could also be significantly reduced using the simplified quantum FMM for evenly distributed particles from Section \ref{sec:even}.
The difficulty in using that scheme is that the quantum simulation includes basis states with unevenly distributed particles, which is why we have presented the full adaptive scheme as our main approach.
The full adaptive scheme has larger overhead due to using more levels, associating multipole information with every electron, and the difficulty of accessing the information from the interaction list.
If it were possible to bound the error while restricting to a subspace of evenly distributed particles, then the simpler scheme of Section \ref{sec:even} could be used.

In the future it may also be interesting to use our scheme as a basis for efficient implementations of the split-operator method for Hamiltonians with more general potentials.
An important special case is the simulation of periodic quantum systems in a dual plane-wave basis \cite{DualPlane}.
That results in the kernel
\begin{equation}
    V(\vb{r}_{i},\vb{r}_{j}) = \frac{4\pi}{\Omega} \sum_\nu \frac{\cos(k_\nu \cdot (\vb{r}_{i}-\vb{r}_{j}))}{k_\nu^2},
\end{equation}
which is quite different from the standard Coulomb interaction in a real-space basis that we have analysed here.
However, adapting FMM to this would involve straightforward modification of the interaction lists and certain translation operators, as in classical implementations~\cite{schmidt1991implementing,schmidt1997multipole,baczewski2011n}.

The fast multipole method is also related to the method of hierarchical matrices \cite{FenSte02}, which allows matrix-vector multiplication of general kernel matrices $\Phi(\Vert \vb{r}_{i}-\vb{r}_{j}\Vert)$ of dimension $n$ in time $\cO(n\,\mathrm{polylog}(n))$.
While related, the fast multipole does not straightforwardly map onto these more general methods.
In particular, hierarchical matrix methods replace the interaction list with \emph{admissible blocks}, which are defined by well-separated pairs of clusters of particles.
However, these clusters are not confined to one level of the hierarchical splitting of the kernel and so require additional arithmetic to avoid overlapping blocks.
The interaction list of a box in our fast multipole scheme  can therefore be understood as a subset of the set of admissible blocks.
Future work is required to include such methods in the context of split-operator-based simulation.

\section*{Acknowledgments}
The authors thank Guang Hao Low, Nathan Wiebe, Rolando Somma, Robin Kothari, and Pedro Costa for helpful discussion. DWB worked on this project under a sponsored research agreement with Google Quantum AI.
DWB is also supported by Australian Research Council Discovery Project DP220101602.
ADB was supported by the DOE Office of Fusion Energy Sciences ``Foundations for quantum simulation of warm dense matter'' project. ECE and AT are supported by Google Quantum AI.
This article has been co-authored by employees of National Technology \& Engineering Solutions of Sandia, LLC under Contract No.\ DE-NA0003525 with the U.S.\ Department of Energy (DOE). 
The authors own all right, title and interest in and to the article and are solely responsible for its contents. 
The United States Government retains and the publisher, by accepting the article for publication, acknowledges that the United States Government retains a non-exclusive, paid-up, irrevocable, world-wide license to publish or reproduce the published form of this article or allow others to do so, for United States Government purposes. 
The DOE will provide public access to these results of federally sponsored research in accordance with the DOE Public Access Plan \url{https://www.energy.gov/downloads/doe-public-access-plan}.

\bibliography{Mendeley,extra}

\appendix
\section{Additional FMM implementation details}
\label{app:additional_fmm}

The main text primarily considers simplified versions of the FMM, Algorithm~\ref{alg:FMM} and Algorithm~\ref{alg:FMMagg} that use zeroth-order expansions.
However, higher-order expansions are essential to achieve arbitrary accuracy in approximating $V$.
In this Appendix, we consider some of the details needed to extend the results in the main text to higher order.

All versions of FMM make use of the fact that the far-field (source-free) contribution to the potential satisfies the Laplace equation.
A generic solution to the Laplace equation can be written as 
\begin{equation}
    \sum \limits_{l=0}^{\infty} \sum \limits_{m=-l}^{l}\left(L_{lm} \|\vb{r}\|^l + M_{lm}\|\vb{r}\|^{-l-1}\right)Y_{lm}(\theta,\phi), \label{eq:generic_expansion}
\end{equation}
where $(r,\theta,\phi)$ are standard spherical coordinates and $Y_{lm}$ are the spherical harmonics. 
Multipole expansions are ones for which the only non-zero coefficients are in $M_{lm}$, and local expansions are ones for which the only non-zero coefficients are in $L_{lm}$.
In practice, either type of expansion is truncated beyond some fixed order $l=\trunc$ and thus comprised of $(\trunc+1)^2$ terms.

Everything in the main text is a special case in which $\mathcal{P}=0$.
In this special case, the multipole expansion coefficient $M_{00}$ for each box is the total charge in that box, and the local expansion coefficient $L_{00}$ for each box is the average far-field contribution to the potential.
That is, $M_{00}(b)$ is equivalent to $Q_b$ in the main text, and $L_{00}(b)$ is equivalent to $V_b$.
Higher-order implementations require machinery for storing and manipulating higher moments of the charge distribution and potential.

Algorithm~\ref{alg:FMM_higher_order} provides the details of an order-$\trunc$ implementation of FMM~\cite{greengard1988efficient}.
As compared to Algorithms~\ref{alg:FMM} and~\ref{alg:FMMagg}, the new elements of this implementation include
\begin{enumerate}
    \item $(\trunc+1)^2$ multipole and local expansion coefficients for each box outside of the leaf level,
    \item operators ($\translate_{MM}$ and $\translate_{LL}$)  that translate the origins of the multipole and local expansions (respectively) from the centroid of one box to the centroid of another, and
    \item operators ($\translate_{ML}$) that translate the multipole expansion coefficients about one box into local expansion coefficients about a box in its interaction list.
\end{enumerate}
The precise details of the various translation operators are one of the primary considerations in an FMM implementation, and they determine the cost of tree traversal. 
We indicate the action of one of the translation operators with a $\circ$, e.g., $\left(\translate_{ML}(b,b')\circ M(b')\right)_{lm}$ translates a multipole expansion about box $b'$ to a local expansion about box $b$.
To remain agnostic of the implementation details particular to different approaches to applying the translation operators, the $\circ$ operation implies some unspecified contraction of the translation operator against its operand to yield an output of the same rank.
Generically, the translation operators only depend on the structure of the tree and the form of the potential ($\Phi$), so they do not require any quantum data access. 
They can be classically precomputed with $\cO(\log N\,\mathrm{polylog}(1/\epsilon))$ overhead in time and memory, which is also typical in purely classical FMM implementations.

In the adaptive version of the quantum FMM in the main text, there is a single grid point per leaf-level box.
Thus, there is no need for additional leaf-level storage for order-$\trunc$ quantum FMM because the monopole moments ($M_{00}$) are the only non-zero moments.
However, other levels of the tree will require $\cO((\trunc+1)^2\log(1/\epsilon))$ qubits per box to store all of the multipole and local expansion coefficients.
The multipole-to-multipole (MM) and local-to-local (LL) translation operators (see loops 2 and 3 of Algorithm~\ref{alg:FMM_higher_order}, respectively) will be applied to generate these coefficients.
Each translation is in 1 of 8 directions to the centroid of its parent (or child), with the distances only depending on which of the $L-1$ pairs of adjacent levels are being traversed.
Similarly, there are at most 189 unique multipole-to-local (ML) translations (see loop 3 of Algorithm~\ref{alg:FMM_higher_order}) in any of the $L-1$ levels in which those translations take place. 
An 8-fold cubic symmetry might also be exploited in implementations in which the ML translation operators strictly depend on the distance between centroids (i.e., any translation can be mapped to an equivalent one in the first octant).
We also note that $L_{00}$ will be the only non-zero component of the local expansion at the leaf level, so an efficient implementation will only use elements of the LL translation operator that contribute to this component.

One of the key design choices in an FMM implementation is the manner in which these translation operators are applied.
A na\"ive classical implementation would require $\widetilde{\cal{O}}(\trunc^4)$ floating-point operations per translation for dense translation operators.
Early on in the history of FMM, it was realised that ML translation operators applied using fast Fourier transforms or in a basis in which they're diagonal would only require $\widetilde{\cal{O}}(\trunc^2)$ floating-point operations~\cite{greengard1988efficient} per ML translation. 
Achieving this scaling in practice required mitigating certain numerical instabilities at the cost of additional implementation complexity~\cite{greengard1997new}.
Alternatively, formulations based on Cartesian tensors would require $\widetilde{\cal{O}}(\trunc^6)$ floating-point operations per ML translation, but could benefit from other design trade-offs (e.g., exact MM and LL translation could be used to reduce the effective size of the interaction list and the number of ML translations)~\cite{shanker2007accelerated}.
It also warrants noting that certain formulations might require complex arithmetic at intermediate levels, even if the resulting potentials are real.
To remain agnostic to these details in the complexity analysis that follows, we will assume that the MM and LL translations both require $\widetilde{\cO}(\trunc^\alpha)$ floating-point operations and the ML translations will require $\widetilde{\cO}(\trunc^\beta)$ floating-point operations, where $\alpha$ and $\beta$ are positive integers with values determined by the precise details of the associated translation operator implementations.
Determining an approach to translation that is best suited to a practical quantum implementation is left to future work.

Finally, we comment on how the choice of $\trunc$ impacts the overall approximation error in $V$.
The approximation error in truncating $V_{i,{\rm far}}$ beyond order $\trunc$ is bounded from above as
\begin{equation}
    |V_{i,{\rm far}}(\trunc\rightarrow \infty)-V_{i,{\rm far}}(\trunc)| \leq \frac{g}{\Omega^{1/3}}\left(\frac{1}{2^{\trunc-L+3}}\right),
\end{equation}
where this bound comes from a routine application of the triangle inequality between leaf-level boxes in each other's interaction lists (i.e., their centroids are separated by at least $\Omega^{1/3}/2^{\Levs-2}$) and $g$ is some $\cO(1)$ constant that depends on the underlying charge distribution.
Thus, it is evident that $\trunc$ should be $\cO(\log(1/\epsilon))$ to ensure an overall approximation error that is $\cO(\epsilon)$.

\begin{algorithm}
    \caption{Classical FMM (order $\trunc$)}\label{alg:FMM_higher_order}
    Initialise $V \leftarrow 0$ \\
    Upward pass\\
    1. \For{$b \in B_{\Levs}$}{
        \For{$l \leftarrow 0$ \KwTo $\trunc$}{
            \For{$m \leftarrow -l$ \KwTo $l$}{
                Initialise $M_{lm}(b) \leftarrow 0$\\
                \For{$j \in p(b)$}{
                    $M_{lm}(b) \leftarrow M_{lm}(b) + q_j \left\|\vb{r}_j-\vb{c}_b\right\|^l Y^{-m}_{l}(\theta_j,\phi_j)$                     
                }
            }
        }
    }
    2. \For{$\ell \leftarrow \Levs-1$ \KwTo $3$}{
        \For{$b \in B_{\ell}$}{
                    Initialise $M_{lm}(b) \leftarrow 0$\\
                    \For{$b' \in C(b)$}{
                        $M_{lm}(b) \leftarrow M_{lm}(b) + \left(\translate_{MM}(b,b') \circ M(b')\right)_{lm}$
                    }
        }
    }
    Downward pass\\
    Initialise $L_{lm}(0),\ldots,L_{lm}(8) \leftarrow 0$ \\
    3. \For{$\ell \leftarrow 3$ \KwTo $\Levs$}{
        \For{$b \in B_{\ell}$}{
                    \For{$b' \in I(b)$}{
                        $L_{lm}(b) \leftarrow L_{lm}(b) + \left(\translate_{ML}(b,b') \circ M(b')\right)_{lm}$
                    }  
                    \If{$\ell\ne L$}{
                     \For{$b' \in C(b)$}{
                        $L_{lm}(b') \leftarrow \left(\translate_{LL}(b,b') \circ L(b)\right)_{lm}$
                    }
                    }
        }
    }
    4. \For{$b \in B_{\Levs}$}{
        \For{$i \in p(b)$}{
            Initialise $V_i \leftarrow 0$ \\
            \For{$l \leftarrow 0$ \KwTo $\trunc$}{
                \For{$m \leftarrow -l$ \KwTo $l$}{
                    $V_i \leftarrow V_i + \left\|\vb{r}_i-\vb{c}_b\right\|^l Y_{l}^{m}(\theta_i,\phi_i) L_{lm}(b)$ 
                }
            }
            \For{$n \in \nearn(b)$}{
                \For{$j \in p(n)$, $i>j$}{
                    $V_i \leftarrow V_i + \Phi(\vb{r}_i,\vb{r}_j) q_j$
                }
            }
            \For{$j \in p(b)$, $i>j$}{
                $V_i \leftarrow V_i + \Phi(\vb{r}_i,\vb{r}_j) q_j$
            }
        $V \leftarrow V+ q_iV_i$
        }
    }
    Return $V$
\end{algorithm}

In the quantum implementation of the full classical FMM, there are a number of simplifications that can be made.
First, in the upward pass at the leaf level, the expansion coefficients are determined using the spherical harmonics $Y_l^{-m}$ in Algorithm \ref{alg:FMM_higher_order}.
In the quantum implementation of the adaptive algorithm, at the leaf level the box size corresponds to the grid resolution.
This means that the particle position is the centre of the box, so $\vb{r}_j-\vb{c}_b=0$, and therefore the only nonzero entries are for $l=m=0$.
At the leaf level, we may simply take $M_{00}(b)=q_j$ as in the monopole case.

Moreover, this simplification may be made when converting $L_{lm}$ to $V_i$ in loop 4.
Because $\vb{r}_j-\vb{c}_b=0$, the only nonzero contribution is for $l=m=0$, and we may set $V_i=L_{00}(b)$.
A further simplification is because in loop 4 we do not need the individual values of $V_i$, and so may add $q_i L_{00}(b)$ to $V$.
Because of the encoding, $L_{00}(b)$ is only stored for charges $q_i=1$, and so we simply add $L_{00}(b)$.
Similarly, for nearest-neighbour boxes, we may simply add $\Phi(\vb{r}_i,\vb{r}_j)$.
The data for these boxes is only present for particles, so the charge need not be explicitly multiplied.
Then the last part in loop 4 is for other particles in the same leaf level box, which may simply be omitted in the quantum implementation.

As compared to the quantum implementation of the monopole algorithm, the key differences to the complexity are as follows.
\begin{enumerate}
    \item The initial complexity in loop 1 of setting $M_{00}(b)$ equal to the charges is trivial.
    \item When combining the multipole information from child boxes in loop 2, we have $\cO(\log N)$ steps where a sort is performed with complexity $\cO(\eta\log\eta)$ steps, but each step has its complexity increased.
    Rather than the items of data to be sorted having size $\cO(\log N)$, we need to include $M_{lm}(b)$.
    That has size $\cO(\trunc^2\log(1/\epsilon))$, given that each entry of $M_{lm}(b)$ is given to $\cO(\log(1/\epsilon)$ bits.
    Taking $\trunc=\cO(\log(1/\epsilon))$, we then replace $\log N$ with $\log^3(1/\epsilon)$.
    Note that these sorts need only include information from $M_{lm}(b)$ for two levels (the parent and child), so we do not incur any extra factor of $\log N$ for the number of levels.
    The complexity of the sorts is therefore
    \begin{equation}\label{eq:potlsorts3}
    \cO(\eta \log \eta \log N \log^3(1/\epsilon)) \, .
    \end{equation}
    \item In order to combine information from child boxes in loop 2, we now need to perform a MM translation to determine the contribution to $M_{lm}(b)$ for the parent box.
    This requires $\widetilde{\cO}(\trunc^\alpha)$ operations, the cost of which will be dominated by multiplications which have complexity $\cO(\log^2(1/\epsilon))$ in terms of elementary gates.
    There are $\cO(\eta\log N)$ steps, and so the complexity is
    \begin{equation}\label{eq:potlMM}
    \widetilde{\cO}(\eta \log N \log^{2+\alpha}(1/\epsilon)) \, .
    \end{equation}
    In this expression, $\widetilde{\cO}$ accounts for logarithmic factors in the MM translation complexity.
    \item When performing the sorts to access the information from boxes in the interaction list, rather than just the box charge we need to include $M_{lm}(b)$ for the given level $\ell$.
    The complexity of this sort is therefore equivalent to that used in loop 2, and is given by Eq.~\eqref{eq:potlsorts3}.
    \item When using Algorithm \ref{alg:copy} to copy the multipole information, $\log\eta$ is replaced with $\log(1/\epsilon)$, and so the complexity becomes
    \begin{equation}
    \cO(\eta \log N \log^3(1/\epsilon)) \, .
    \end{equation}
    \item Rather than just adding into the potential, we need the ML translation, to give the contribution to $L_{lm}(b)$ from $M_{lm}(b')$ for the box in the interaction list.
    This requires $\widetilde{\cO}(\trunc^{\beta})$ operations that will also be dominated by multiplications, and so Eq.~\eqref{eq:mults} is replaced with
    \begin{equation}\label{eq:mults2}
    \widetilde{\cO}(\eta \log N \log^{2+\beta}(1/\epsilon)) \, . 
    \end{equation}
    \item A task which had trivial complexity in the monopole version was copying the local data to the child boxes (because it was assumed the potential was unchanged).
    Here, that now requires LL translations, which has the same complexity as the MM translations above.
    Therefore, the complexity is again given by Eq.~\eqref{eq:potlMM}.
    \item The final additions into the total potential energy in loop 4 are essentially unchanged, because we are adding $L_{00}(b)$ as well as $\Phi(\vb{r}_i,\vb{r}_j)$, which is given by QROM.
    Therefore the complexity is
    \begin{equation}
    \cO(\eta\log(1/\epsilon)) \, . 
    \end{equation}
\end{enumerate}

The dominant contributions to the complexity are from the translations, which will have complexity
\begin{equation}
    \widetilde{\cO}(\eta\trunc^{\max(\alpha,\beta)}\log N  \log^2(1/\epsilon)) = \widetilde{\cO}(\eta\log N \log^{2+\mathrm{max}(\alpha,\beta)}(1/\epsilon)).
\end{equation}
Here, the right expression is from taking the truncation level to be $\cO(\log(1/\epsilon))$.
There are also complexities with a factor of $\log\eta$ for sorts, but that complexity should be smaller because $\eta$ should be smaller than $1/\epsilon$.
It is needed to store data for $M_{lm}$ and $L_{lm}$ at each level and for each particle, and so the qubit usage is
\begin{equation}\label{eq:ho_memory}
    {\cO}(\eta\trunc^2\log N  \log(1/\epsilon)) = {\cO}(\eta\log N \log^3(1/\epsilon)). 
\end{equation}
We give further details of these costings with example parameters in Appendix \ref{app:breakeven}.

We note another caveat in adapting existing classical implementations of FMM to a quantum setting.
We consider one of the canonical FMM implementations described in Ref.~\cite{greengard1997new}.
There, the cost of tree traversal is dominated by the MM and LL translations for which $\alpha=3$, whereas $\beta=2$ for the ML translations.
Nevertheless, the overall scaling for tree traversal is $\cO(\trunc^2)$ (not $\cO(\trunc^3)$) because the number of translations is taken to scale inversely with $\trunc$.
This is a common trick in many implementations that is accomplished by letting the leaf-level particle density scale with $\trunc$, effectively shortening the tree.
While it is still possible to shorten the tree by scaling the number of grid points per leaf box with $\trunc$ in our quantum implementation, this only impacts time and memory complexities by $\cO(\log(\trunc))$ factors, i.e., substituting $N\rightarrow N/\trunc$ in Eqs.~\eqref{eq:potlsorts3}-\eqref{eq:ho_memory}.
Thus, a straightforward adaptation of the scheme in Ref.~\cite{greengard1997new} will instead have $\cO(\trunc^3)$ cost and we expect that any classical implementation that achieves better asymptotics by adjusting the leaf-level density will suffer a similar drawback.
Still, adapting the approach from Ref.~\cite{greengard1997new} is asymptotically superior to the $\cO(\trunc^4)$ cost that would be realised by a totally naive approach to translation.
But borrowing well-established tricks from the classical literature might not perform as well as expected.

Another compelling approach is to avoid ML and LL translations entirely, evaluating the far-field contribution to the potential at each particle directly from multipole expansions aggregated by MM translation alone.
This is described in Appendix~\ref{app:M2P}.
Regardless of the MM translation scheme used, the far-field potential evaluation from the multipole moments requires $\cO(\eta\trunc^2\log N \log^2(1/\epsilon))$ operations.
However, this requires the use of coherent arithmetic to evaluate solid harmonics in lieu of using QROM to access precomputed ML and LL translation operators.
Nevertheless, we will show that this appears to be superior to a straightforward quantum implementation of one of the canonical classical implementations~\cite{greengard1997new}.

\section{Break-even point for speedup over direct summation}
\label{app:breakeven}
In this Appendix, we estimate the constant factors required to determine the break-even point beyond which the FMM will provide a speedup over direct pairwise summation.
This is heavily implementation dependent, so we will consider a quantum version of Greengard and Rokhlin's implementation in Ref.~\cite{greengard1997new}.
This requires cost models for the number of operations involved in direct pairwise summation, $\cC_{\textrm{direct}}(\epsilon)\eta^2$, and the FMM implementation, $\cC_{\textrm{tree}}(\epsilon) \eta$.
Here, $\cC_{\textrm{direct}}$ is the cost of evaluating the Coulomb potential for a single pair of particles and $\cC_{\textrm{tree}}$ is the cost of tree traversal in the FMM.
We are concerned with identifying the value of $\eta$ beyond which $\cC_{\textrm{direct}}(\epsilon)\eta^2>\cC_{\textrm{tree}}(\epsilon) \eta$ for a fixed value of $\epsilon$.
This is
\begin{equation}
    \eta_{\textrm{BE}}(\epsilon) = \mathcal{C}_{\textrm{tree}}(\epsilon)/\mathcal{C}_{\textrm{direct}}(\epsilon). \label{eq:breakeven}
\end{equation}
The cost of evaluating the Coulomb potential for a single pair of particles is $\cO(\log^2(1/\epsilon))$~\cite{arithmetic}.
The cost of tree traversal is primarily determined by the scaling with $\trunc \sim \cO(\log(1/\epsilon))$, which is implementation dependent as discussed in Appendix~\ref{app:additional_fmm}.

Data from Greengard and Rokhlin's classical implementation are consistent with $\trunc \approx 0.8 \log_2(1/\epsilon)$~\cite{greengard1997new}.
While we should expect that an equivalent quantum implementation will make use of a taller tree, this should not dramatically change the constant factor coefficient to $\log_2(1/\epsilon)$ because $L$ is also $\cO(\log(1/\epsilon))$.
However, we will see that this can have a significant impact on the cost of data movement in the quantum implementation.
Greengard and Rokhlin also report the following data for $\eta_{\textrm{BE}}(\epsilon)$:
\begin{itemize}
    \item $\eta_{\textrm{BE}}(\epsilon=4.5 \times 10^{-3}) = 500$, for $\trunc=5$.
    \item $\eta_{\textrm{BE}}(\epsilon=1.4 \times 10^{-4}) = 2000$, for $\trunc=9$.
    \item $\eta_{\textrm{BE}}(\epsilon=1.1 \times 10^{-7}) = 4000$, for $\trunc=18$.
    \item $\eta_{\textrm{BE}}(\epsilon=6.2 \times 10^{-12}) = 5000$, for $\trunc=30$.    
\end{itemize}
This does not necessarily translate directly to the break-even point for a quantum implementation because of the additional costs associated with retrieving data for the interaction list during tree traversal.
Because the classical implementation can make use of fast access to a precomputed octree during tree traversal, it is important to consider the relative cost of data movement and the coherent arithmetic involved in $\cC_{\textrm{tree}}(\epsilon)$ before extrapolating from the classical implementation.

\subsection{Data retrieval costs}
\label{app:data_movement}

There are additional costs involved in retrieving the data in the quantum implementation.
\begin{itemize}
    \item For 3D simulations, there are 8 sorts used to retrieve the data from the interaction list.
    \item For each of those sorts we run through the list of particles once, and account for particles that may be shifted by up to 63 boxes in 3D.
\end{itemize}
To compare the complexity for these operations to the complexity for the operations, we will consider the ratio of the gate complexity to the size of the data for the multipole information.
Here we are considering only a single level at a time, and there is multipole information for each electron of size approximately $\trunc^2\log_2(1/\epsilon)$.

For the sorts, the complexity is $\cO(\eta\log \eta)$ comparators, and for small $\eta\le 32$ highly optimised sorting networks are known.
For larger $\eta$ of thousands a systematic construction is needed, and the $\cO(\eta\log \eta)$ constructions have very large constants.
For example the Zig-zag sort of Goodrich has a best-case complexity of $2700\eta\log_2\eta$ \cite{Zigzag}.
In comparison, Batcher's odd-even mergesort \cite{Batcher} has complexity
\begin{equation}
    \tfrac 14 \eta\log_2^2\eta - \tfrac 14 \eta\log_2\eta + \eta - 1,
\end{equation}
giving complexity approximately $34\eta$ for $\eta=4096$.
The theoretical lower bound is $\lceil\log_2(\eta!)\rceil$, which gives about $10.6\eta$ for $\eta=4096$, so improved sorting networks could reduce the complexity by at most a factor of 3.

For these sorts there are controlled swaps of the multipole information for each comparator.
We also need to move the information for the local expansion, which is the same size as the multipole information.
There is another factor of 2 in order to invert the sorts.
As a result, the combined complexity/data ratio for the sorts is about $2\times 2\times 8\times 34=1088$.

When running through the list of particles, the primary cost is copying data into appropriately shifted registers using Algorithm \ref{alg:copy}.
These registers are used for multipole information from nearby boxes, so it is possible to retrieve information from the interaction list.
In the following discussion, we will call these registers the partial interaction (PI) list as they only include information from part of the interaction list.
As written in Algorithm \ref{alg:copy}, the value of $\Delta k$ is calculated, then data is copied using $\Delta k$ to give the data location.
The extra overhead of $K=63$ for coherent data access can be avoided by using a cyclic shift, as for example given in Ref.~\cite{Niroula2021}.
That is, the bits of $\Delta k$ are used to control a series of swaps to move the data by $\Delta k$ positions.

The procedure can be further simplified by using a single PI list, rather than copying and shifting it for the next electron.
The method is to include the position and multipole information for electron $j$ with the PI list, compute $\Delta k$, then use it to control the cyclic shift of these 64 items, including the information for electron $j$.
Items 2 to 64 are then used for the PI list for electron $j+1$.
If electron $j+1$ is in the same box as $j$, then there is no shift, and the multipole information for electron $j$ remains in that data location.
If there is a shift, then the multipole information for electron $j$ is replaced with information from the PI list that is not needed for electron $j+1$.

In this way there is no cost of copying to a new list, and also only one copy of the PI list need be retained, rather than needing copies of the PI list for each electron.
The values of $\Delta k$ need be kept for each electron to invert the procedure, but only 6 qubits are needed for each value of $\Delta k$, making it a negligible qubit cost.
Because the PI list is updated for each electron, the multipole-to-local translations to update the local expansion must be performed before proceeding to the next electron.

There is an overhead of $6$ for the cost of the procedure shifting the entries in the PI list, corresponding to the number of bits of $\Delta k$.
Therefore, for copying the neighbour information, there is a total complexity/data ratio of $2\times 8\times 6\times 64=6144$, where the factor of $2$ is for inverting the procedure, $8$ is for the different sorts, $6$ is for the controlled shift, and $64$ is the size of the PI list being shifted.

So far we have just considered the movement of the multipole and local expansion.
As will be discussed in the next subsection, the gate complexity is reduced by computing outgoing exponential expansions from the multipole information.
The best gate complexity is obtained by calculating these for all boxes first, rather than calculating every time a box in the interaction list needs to be accessed.
This means that the outgoing exponential expansions would need to be moved in the sort and in shifting the entries in the PI list, rather than the multipole information.

The size of the data for the exponential expansions is approximately $2\times 6\times P_{\textrm{exp}}(\epsilon)\log_2(1/\epsilon)$, where the factor of 2 is for complex numbers, 6 is for the 6 faces of a cube, and $P_{\textrm{exp}}(\epsilon)\approx 1.72\trunc^2$ is the number of basis functions.
This results in the size of the data for the exponential expansion being about a factor of $18.5$ times larger than the data for the multipole information (with $\trunc=18$).
This results in the complexity/data ratio of 1088 above for the sorts becoming 10,621, and the ratio of 6144 for the PI list becoming 113,815.

\subsection{Coherent arithmetic costs}
\label{app:coherent_arithmetic}

The costs of the arithmetic involved in a quantum implementation of the tree traversal scheme in Ref.~\cite{greengard1997new} are described in what follows.
Reference~\cite{greengard1997new} considers several approaches to tree traversal: (1) one used in the original FMM~\cite{greengard1988efficient} with straightforward translations, (2) a refinement in which the cost of translations is reduced through the use of rotation matrices, and (3) a further refinement that uses exponential expansions to reduce the cost even more.
Approach (3) has the best scaling in $\trunc$ and it is the method that was used in generating the break-even points above, so it is what we consider here.

Computing the $\ell=L$  multipole expansions from the charges and the potential from the $\ell=L$ local expansions has negligible cost, because only the monopole moments contribute.
All other MM and LL translations require $\cO(\trunc^3)$ operations and the ML translations require $\cO(\trunc^2)$ operations.
The MM and LL translations are reduced from $\cO(\trunc^4)$ operations due to the use of rotation matrices that orient the expansions such that each translation involves two sets of $\cO(\trunc^2)$ multiplications and three sets of $\cO(\trunc^3)$ multiplications.
The ML translations are more complicated. 
They require converting each multipole expansion into an outgoing exponential expansion, translating in that exponential basis, and converting the resulting incoming exponential expansion into a local expansion.
Each outgoing and incoming exponential expansion is itself decomposed into six distinct directions (one for each face of a box) and a quadrature rule with $\cO(\trunc^2)$ points is used to evaluate an integral representation of the Coulomb potential, implementing translation in the exponential basis.
It should be noted that the ML interaction list for any given box also needs to account for which of the six outgoing and six incoming expansions are properly oriented, i.e., an outgoing exponential expansion leaving one face will only be ``received'' by incoming exponential expansions on oppositely oriented faces. 
Overall, the multipole-to-exponential and exponential-to-local expansions each involve $\cO(\trunc^3)$ multiplications, which only need to be performed once per box.
Each translation in the exponential basis requires $\cO(\trunc^2)$ multiplications and there are at most 189 such translations per box.

We proceed to count the number of multiplications \emph{per particle} during tree traversal.
It is impractical to account for all of the implementation optimisations that are described in Ref.~\cite{greengard1997new}, but we expect to be within an order of magnitude.
For the purposes of estimating the overhead over a classical implementation, this should suffice.
The number of levels where this full summation is performed is $L-3$, so we give that factor in the following estimates.
There is also a direct summation at the leaf level which will be considerably lower cost than the multipole calculations.
Level $\ell=3$ of the multipole calculation involves fewer boxes in the interaction list, but we omit that cost saving for simplicity.
\begin{enumerate}
    \item Each MM translation requires two rotations about the $z$ axis ($\trunc(\trunc+1)$ multiplications each), two rotations about the $y$ axis ($(\trunc+1)(2\trunc+1)(2\trunc+3)/3$ multiplications each), and one translation along the $z$ axis ($(\trunc+1)(\trunc+2)(4\trunc+3)/6$ multiplications). 
    Thus, there are $\left(\frac{10}{3}\trunc^3+\frac{25}{2}\trunc^2+\frac{73}{6}\trunc+3\right)(L-3)$ multiplications per particle due to MM translations up the tree.
    \item Each real multipole expansion will be converted into six complex outgoing exponential expansions.
    Each exponential expansion is given in terms of $P_{\textrm{exp}}(\epsilon)$ basis functions, where $P_{\textrm{exp}}(\epsilon) \sim \cO(\trunc^2)$ is the number of quadrature points used to discretise an integral representation of the Coulomb kernel.
    For accuracy consistent with $\trunc=18$ (i.e., $\epsilon=1.1 \times 10^{-7}$), $P_{\textrm{exp}}(\epsilon)=558$ in Ref.~\cite{greengard1997new},
    which implies $P_{\textrm{exp}}(\epsilon)\approx 1.72 \trunc^2$.
    The number of multiplications involved in converting the multipole expansion to each of the six outgoing exponential expansions will vary, due to the need to apply different numbers of rotation matrices to align with different directions of propagation.
    \begin{enumerate}
        \item There is one conversion that does not involve any rotations, and thus it requires $\trunc(\trunc+1)^2 +2\trunc P_{\textrm{exp}}(\epsilon)$ multiplications.
        The factor of two in the second term accounts for the fact that each operation is actually a multiplication between a real number and a complex number.
        In this case and henceforth, the multiplications are ordered to minimise the number of complex multiplications.
        \item There are three conversions that also involve $y$-axis rotations, requiring $(\trunc+1)(2\trunc+1)(2\trunc+3)/3+\trunc(\trunc+1)^2 +2\trunc P_{\textrm{exp}}(\epsilon)$ multiplications each.
        \item There are two conversions that involve both $y$- and $z$-axis rotations, with $\trunc(\trunc+1)+(\trunc+1)(2\trunc+1)(2\trunc+3)/3+\trunc(\trunc+1)^2 +2\trunc P_{\textrm{exp}}(\epsilon)$ multiplications each.
    \end{enumerate}
    Thus, there are $\left(12P_{\textrm{exp}}(\epsilon)\trunc+\frac{38}{3}\trunc^3+34\trunc^2+\frac{79}{3}\trunc+5\right) (L-3)$ multiplications per particle for the multipole-to-exponential conversions.
    \item Each outgoing-to-incoming ML translation requires $3P_{\textrm{exp}}(\epsilon)$ multiplications, where the factor of 3 reflects the fact that these are multiplications between complex numbers.
    There are at most 189 interacting pairs that require ML translation for any given box and 6 pairs of outgoing/incoming exponential expansions per box.
    Thus, there are at most $3402P_{\textrm{exp}}(\epsilon)(L-3)$ multiplications involved in the outgoing-to-incoming translations, across the tree.
    \item Conversion of the six incoming exponential expansions mirrors the cost of the conversion of the multipole expansions into outgoing exponential expansions. 
    Thus, there are $\left(12P_{\textrm{exp}}(\epsilon)\trunc+\frac{38}{3}\trunc^3+34\trunc^2+\frac{79}{3}\trunc+5\right) (L-3)$ multiplications per particle for the exponential to LL conversions.
    We note that the data become purely real again during this step.
    \item Each LL translation has the same cost as an MM translation. 
    Thus, there are $\left(\frac{10}{3}\trunc^3+\frac{25}{2}\trunc^2+\frac{73}{6}\trunc+3\right)\left(L-3\right)$ multiplications per particle for the LL translations down the tree.    
\end{enumerate}

The total number of multiplications per particle for tree traversal then becomes
\begin{equation}\label{eq:classFMM}
    \left(\left(3\times6\times189 +24\trunc\right) P_{\textrm{exp}}(\epsilon)+32\trunc^3+93\trunc^2+77\trunc+16\right)(L-3) \approx \left(73\trunc^3+5944\trunc^2+\cO(\trunc)\right)(L-3).
\end{equation}
The quantum implementation is more costly than this, because rather than directly accessing all boxes in the interaction list, we are accessing the boxes through 8 different sorts.
In each of these sorts we retrieve 63 neighbours, though at most 56 may be in the interaction list.
This means that the outgoing-to-incoming ML calculations in Step 3 need to be performed 448 times, rather than 189.
A further overhead is that the ML translation in step 3 should be uncomputed before the next step, to avoid needing to store these intermediate results of the calculation in ancillae.
As a result, the 189 in Eq.~\eqref{eq:classFMM} should be replaced with $448\times 2$.

A further difference is that the incoming exponential expansions should be converted for each of the 8 sorts.
That is because the data for the incoming exponential expansions needs to be uncomputed to erase the ancillae this information is stored in, and
that uncomputation needs to be performed before the next sort.
(Otherwise the incoming exponential expansions would need to be stored for every particle for every level, greatly increasing the data storage requirements.)
As a result there is a factor of 16 on the complexity of step 4.
We also uncompute the exponential expansion in step 2, giving a further factor of 2 for that step.
Combining these considerations gives the complexity for the quantum tree traversal as
\begin{equation}\label{eq:quantFMM}
    \left(\left(3\times6\times 896 +216\trunc\right) P_{\textrm{exp}}(\epsilon)+\frac{704}3 \trunc^3+637\trunc^2+\frac{1495}3\trunc+96\right)(L-3) \approx \left(606\trunc^3+28377\trunc^2+\cO(\trunc)\right)(L-3).
\end{equation}

Considering the ratio of the complexity to the size of the multipole data, there is an extra factor of $\log_2(1/\epsilon)$ for translating multiplications to elementary gates.
There would be an ratio of around $(73 \times 18 + 5944)\times 22 =159{,}676$ over the size of the data for coherent arithmetic with 22 bits, were it possible to work directly with the interaction list.
Accounting for the retrieval of data using sorts, as well as uncomputation costs, the ratio becomes $(606 \times 18 + 37624)\times 22 = 864{,}274$.
Including the overhead for the sorts and data retrieval, this means that the coherent implementation has a factor of about $5.4$ larger complexity than the classical implementation (as estimated in Eq.~\eqref{eq:classFMM}).

It is also possible to compute the outgoing exponential expansions (step 2) when building the PI list.
That is, every time multipole information is added to the PI list, the outgoing exponential expansion is calculated and added to the PI list as well.
Then when multipole information is removed from the PI list the corresponding outgoing exponential expansion is uncomputed.
This greatly reduces the amount of data that need be stored, but the cost of step 2 is multiplied by 16 (due to the 8 sorts) instead of 2, giving the complexity for the quantum tree traversal as
\begin{equation}\label{eq:quantFMM2}
    \left(\left(3\times6\times 896 +384\trunc\right) P_{\textrm{exp}}(\epsilon)+412 \trunc^3+1113\trunc^2+867\trunc+166\right)(L-3) \approx \left(1072\trunc^3+28853\trunc^2+\cO(\trunc)\right)(L-3).
\end{equation}
This costing results in a complexity/data ratio of about $10^6$, or about 6.6 times that in Eq.~\eqref{eq:classFMM}.

\subsection{Other additional costs}
\label{app:other_costs}

Another factor that may cause the quantum implementation to be more costly than that given in Ref.~\cite{greengard1997new} is that it uses more levels.
We consider a number of levels sufficient to resolve to the grid points, so for example $7$ levels for $2^{21}$ grid points.
In comparison, the number of levels given in Ref.~\cite{greengard1997new} was about half that.
There will also be an overhead because we have associated multipole information with each electron.
In practice much of the advantage of the FMM is because many electrons are grouped together in each box, greatly reducing the number of boxes below the number of particles.
Together with the factor of about 5 due other quantum overheads from the above subsection, the threshold value of $\eta$ where an improvement is obtained would be increased by at least an order of magnitude.

Another issue that should be considered in the practicality of the quantum FMM is the number of qubits used.
To give a rough estimate, we will consider the example of order-18 FMM, 22 bits of precision, $\eta=4000$ electrons, and $N=2^{21}$.
Given an FMM order of 18 and 22 bits of precision, there would be about $8000$ qubits for the multipole information for each electron.
With $N=2^{21}$ there would be 5 levels to store multipole information for, giving a total of about 160 million qubits for the multipole expansion.
There will also be a similar amount of information for the local expansion, which would mean about 320 million qubits.

Note that we need only include storage costs for multipole information at each level, whereas Section \ref{sec:calcmult} presented a simplified scheme which involves calculation of charge / multipole information at intermediate levels.
As noted in Section \ref{sec:calcpotl}, one may use Algorithm \ref{alg:copy} to retrieve information from all 8 child boxes to compute multipole information for a parent box, so intermediate levels are not needed.
Similar to the PI list, a single list should be used, rather than a separate list for each electron.

In the above implementation, we first consider the case that the outgoing exponential expansions are stored at each level as well, but only one level at a time before being uncomputed.
That would need about 150,000 qubits for each electron, making a total of 590 million qubits.
The total qubit usage would then be about 910 million.
Alternatively, we may just store the outgoing exponential expansions for the PI list, and uncompute them when they are removed from the list.
That reduces the number of qubits for the outgoing exponential expansions to about 9 million, for a total of about 330 million qubits, though it would increase the gate cost.

\subsection{Multipole-to-particle evaluation}
\label{app:M2P}
Given the data usage and overheads involved with using the full FMM with ML and LL translations, it may be preferable to perform a direct multipole-to-particle calculation of the potential energy.
Here we show that this approach significantly reduces both the qubit and gate cost.
The method is to evaluate the contribution to the potential energy at particle $i$ due to its interaction with particles in box $a$ as
\begin{equation}
    V_{i,a} = q_i \sum \limits_{l=0}^{\trunc} \sum \limits_{m=-l}^l M_{lm}(a) \|\vb{r}_i-\vb{c}_a\|^{-l-1} Y_{lm}(\theta_{i,a},\phi_{i,a}), \label{eq:multipole-particle-eval}
\end{equation}
without the use of a local expansion.
This is particularly useful in the quantum implementation, because we perform the calculation for the multipole information associated with each electron.
A major advantage of this approach is that it is no longer necessary to store the outgoing expansion information for all particles, which greatly increases the data storage requirements in the full method.
Moreover, it is no longer necessary to store the local expansions.
The majority of the data storage requirement is then just the multipole information, which reduces the qubit requirement to about 160 million for 4000 electrons.
There is then the possibility for compressed representations to reduce the data requirement further \cite{Tucker,Hesford2011}.

It is relatively straightforward to evaluate each term in Eq.~\eqref{eq:multipole-particle-eval}. 
The irregular solid harmonics may be computed in real form using the following recurrence relations~\cite{perezsaborid2008coordinate}.
In the following we denote the real and imaginary components of $r^{-l-1}Y_{lm}$ by $C_{lm}$ and $S_{lm}$ respectively, with $x,y,z$ being the components of the vector of length $r$.
\begin{enumerate}
\item
The recursion starts with the monopole term for $l=0, m=0$:
\begin{equation}
C_{0,0} = \frac{1}{r}, \qquad S_{0,0} = 0.
\end{equation}

\item Then the terms along the diagonal for $m = l > 0$
are computed using
\begin{align}
C_{l,l} &= \frac{2l-1}{r^2} \left( x C_{l-1,l-1} - y S_{l-1,l-1} \right), \\
S_{l,l} &= \frac{2l-1}{r^2} \left( x S_{l-1,l-1} + y C_{l-1,l-1} \right).
\end{align}

\item The terms in the off-diagonal for $m=l -1$ are given by
\begin{align}
C_{l,l-1} &= \frac{(2l-1) z}{r^2} C_{l-1,l-1}, \\
S_{l,l-1} &= \frac{(2l-1) z}{r^2} S_{l-1,l-1}.
\end{align}

\item More generally, for $0 \leq m <l- 1$ the terms are calculated using the three-term recurrence
\begin{align}
C_{l,m} &= \frac{(2l-1) z}{r^2} C_{l-1,m} - \frac{(l+m-1)(l-m-1)}{r^2} C_{l-2,m}, \\
S_{l,m} &= \frac{(2l-1) z}{r^2} S_{l-1,m} - \frac{(l+m-1)(l-m-1)}{r^2} S_{l-2,m}.
\end{align}
\end{enumerate}
The values of $C_{l,m},S_{l,m}$ for negative values of $m$ are related to their positive counterparts up to a sign, so do not need to be computed or stored.

The fast inverse square root \cite{arithmetic} can be used to compute $1/r=\|\vb{r}_i-\vb{c}_a\|^{-1}$ with complexity comparable to four multiplications \cite{RubinPNAS24}.
In turn, that can be used to compute $1/r^2$, $x/r^2$, $y/r^2$, and $z/r^2$ with four multiplications.
Given that these quantities have been calculated, and ignoring for the moment the multiplication by classically chosen numbers such as $2l-1$, we have the following accounting of the number of multiplications.
\begin{itemize}
    \item Step 2 above is performed $\trunc$ times, and each has two multiplications, for $2\trunc$ multiplications.
    \item Step 3 above is also performed $\trunc$ times and has two multiplications, for a total of $2\trunc$ multiplications.
    \item Step 4 above is performed $\trunc(\trunc-1)/2$ times, and each has four multiplications, for a total of $2\trunc(\trunc-1)$.
\end{itemize}
Together there is a total of $2\trunc(\trunc+1)$ multiplications.

The multiplications by $2l-1$ and $(l+m-1)(l-m-1)$ correspond to what is known as the Multiple Constant Multiplication problem \cite{Multiplierless}, and cost only a single addition each.
For example, for $l=2$, $2l-1=3=2+1$, so multiplication by $2l-1$ can be obtained by one addition between the multiplicand and the bit-shifted multiplicand.
Then for $l=4$, $2l-1=7=4+3$, so the result can be obtained from an addition between the previous result and the bit-shifted multiplicand.
Similarly, all other results can be obtained either by a single addition directly or a single addition between one of the other results and the bit-shifted multiplicand.
For the multiplications by $(l+m-1)(l-m-1)$, most can be obtained in this way, but others can be obtained by bit-shifted copies of other products.
Therefore, we can ignore this cost in comparison to the other products.

For the total cost of the MP procedure, we lastly need to multiply by the coefficients $M_{lm}(a)$, which gives another $(\trunc+1)^2$ multiplications.
As a result, the total number of multiplications is $3\trunc^2+4\trunc+9$.
In this case there is \emph{no} extra factor of 2 for uncomputation, because the contribution to the potential may be used to apply a phase factor for the simulation, and does not need to be added into a total potential.
Multiplying by a factor of 448 for the boxes in the interaction list with sorts gives $1344\trunc^2+1792\trunc+4032$.
In this approach the MM translation is still needed, but not the LL translation, giving the number of multiplications per particle
\begin{equation}
    \left(1344\trunc^2+1792\trunc + 4032+\frac{10}{3}\trunc^3+\frac{25}{2}\trunc^2+\frac{73}{6}\trunc+3\right)(L-3) \approx
   \left(\frac{10}{3}\trunc^3+1356\trunc^2+1804\trunc + 4035\right)(L-3) .
\end{equation}
For $\trunc=18$, this is less than $1/20$ the multiplication count of the quantum FMM in Eq.~\eqref{eq:quantFMM} and even less than 1/4 the multiplication count of the straightforward translation of the Greengard and Rokhlin algorithm given in Eq.~\eqref{eq:classFMM}.

For the example of $L=5$, this would give about 660,000 multiplications per particle for $\trunc=9$, or about $2.5$ million for $\trunc=18$.
In comparison, the explicit summation needs approximately 4 multiplications for each of $\eta(\eta-1)/2$ pairwise potentials, for a total of about $2\eta$ multiplications per particle.
This would give a threshold for speedup over direct summation of about 330,000 particles for $\trunc=9$, or $1.2$ million for $\trunc=18$.
These thresholds are far greater than estimated for the classical FMM, indicating that further optimisation of the quantum implementation would be needed for it to be practical.

In practice, in FMM there is a considerable reduction in the complexity due to using far fewer boxes than particles, which results in it being more efficient than performing calculations using the MP.
It may therefore be expected that to achieve performance comparable to the classical FMM, the implementation of the quantum FMM should be modified to only represent information for the maximum number of boxes needed.
It may also be advantageous to develop a hybrid scheme with direct summation, or the 
quantum FMM for evenly distributed particles.

\end{document}